%% file: main_arxiv.tex
\documentclass[11pt]{article}
\usepackage{arxiv}
\usepackage{iftex}
\ifXeTeX
\usepackage{fontspec}
\else
\usepackage[T1]{fontenc}
\usepackage{tgtermes,tgheros}
\fi
\usepackage{graphicx}
\usepackage{amsmath,amssymb,amsthm,bm}
\theoremstyle{plain}
\newtheorem{proposition}{Proposition}
\theoremstyle{remark}

\usepackage{booktabs}
\usepackage{adjustbox}
\newcommand{\tabinput}[1]{\IfFileExists{tables/#1.tex}{\begin{adjustbox}{max width=\linewidth}\input{tables/#1}\end{adjustbox}}{(table #1 pending)}}
\usepackage{multirow}
\usepackage{xcolor}
\usepackage[authoryear,round]{natbib}
\usepackage{setspace}
\usepackage[hidelinks]{hyperref}
\usepackage{orcidlink} 
\usepackage{placeins}
\newcommand{\sref}[1]{\ref{#1}}
\newcommand{\mainof}{}
\usepackage[capitalise]{cleveref}
\makeatletter
\newcommand{\mref}[1]{\mref@split#1\@nil}
\def\mref@split#1:#2\@nil{\ifcsname mref@name@#1\endcsname\csname mref@name@#1\endcsname{\ref{#1:#2}}\else\cref{#1:#2}\fi}
\newcommand{\mref@name@fig}[1]{Figure~#1}
\newcommand{\mref@name@tab}[1]{Table~#1}
\makeatother
\input{macros}

\title{Bayesian group testing regression with a shape-free dilution curve for pools of mixed sizes}
\author{%
  Chun-Hao Yang\,\orcidlink{0000-0002-2522-5957}\thanks{Corresponding author: \texttt{chunhaoy@ntu.edu.tw}} \\
  Institute of Statistics and Data Science \\ National Taiwan University \\ Taipei, Taiwan
  \And
  Wei-Yan Hong \\
  Institute of Statistics and Data Science \\ National Taiwan University \\ Taipei, Taiwan
}
\date{}
\renewcommand{\shorttitle}{Shape-free dilution curve for pools of mixed sizes}

\begin{document}
\maketitle

\begin{abstract}
\input{sections/00-abstract}

\end{abstract}

\input{sections/01-introduction}
\input{sections/02-model}
\input{sections/03-bayes}
\input{sections/04-simulation}
\input{sections/05-application}
\input{sections/06-conclusion}

\section*{Acknowledgements}
The authors thank the National Center for Health Statistics for making the NHANES hepatitis~E serology data publicly available.

\section*{Data Availability Statement}
The NHANES 2017--2018 data used in \cref{sec:app} are publicly available from the National Center for Health Statistics (\url{https://wwwn.cdc.gov/nchs/nhanes/}). R code reproducing every table and figure is provided as Supplementary Material.

\section*{Supplementary Material}
Supplement~A gives the derivations; Supplement~B gives checks of the sampler and the Bayes factor, the full simulation results and the sensitivity analyses referenced in \cref{sec:bayes,sec:sim,sec:app}.

\bibliographystyle{plainnat}
\bibliography{refs}

\clearpage
\appendix
\renewcommand{\thesection}{Supplement \Alph{section}}
\renewcommand{\thesubsection}{\Alph{section}.\arabic{subsection}}
\renewcommand{\theequation}{A.\arabic{equation}}
\renewcommand{\tablename}{Supplementary Table}\renewcommand{\thetable}{\arabic{table}}
\renewcommand{\figurename}{Supplementary Figure}\renewcommand{\thefigure}{\arabic{figure}}
\crefname{table}{Supplementary Table}{Supplementary Tables}
\crefname{figure}{Supplementary Figure}{Supplementary Figures}
\setcounter{equation}{0}\setcounter{table}{0}\setcounter{figure}{0}
\renewcommand{\theHtable}{S\arabic{table}}\renewcommand{\theHfigure}{S\arabic{figure}}\renewcommand{\theHequation}{S\arabic{equation}}
\let\origsection\section\renewcommand{\section}{\FloatBarrier\origsection}
\let\origsubsection\subsection\renewcommand{\subsection}{\FloatBarrier\origsubsection}
\input{sections/appendix}

\FloatBarrier

\end{document}

%% file: macros.tex
\newcommand{\Se}{\mathrm{Se}}
\newcommand{\Sp}{\mathrm{Sp}}
\newcommand{\Prob}{\mathrm{pr}}

\newcommand{\tp}{^{\mathsf{T}}}

\newcommand{\R}{\mathbb{R}}
\newcommand{\E}{\mathbb{E}}

\newcommand{\Ytil}{\widetilde{Y}}

\newcommand{\bbeta}{\bm{\beta}}

\newcommand{\btheta}{\bm{\theta}}
\newcommand{\bu}{\bm{u}}

\newcommand{\by}{\bm{y}}

\newcommand{\bX}{\bm{X}}
\newcommand{\bY}{\bm{Y}}
\newcommand{\bZ}{\bm{Z}}



%% file: sections/00-abstract.tex
Pooling specimens cuts screening costs, but dilution by negatives lowers pooled sensitivity, which then depends on the number of positives and the pool size. 
We propose BADGER (\textbf{B}ayesian \textbf{A}nalysis of \textbf{D}ilution in \textbf{G}roup t\textbf{E}sting \textbf{R}egression), in which pooled sensitivity is a nondecreasing function of the fraction of positives, of unknown shape and shared across pool sizes, with a Dirichlet mixture prior. 
A latent-cause augmentation gives conjugate updates for the curve and the assay parameters, and a Bayes factor tests for dilution. 
We show that pooled data cannot separate a higher prevalence from a lower pooled sensitivity: covariate effects are unaffected, but the intercept is learned only through the ends and monotonicity of the curve, fractions shared across pool sizes and covariate differences between pools. 
In simulations, BADGER estimated the curve as well as or better than parametric submodels in most settings with dilution and removed most of the intercept bias caused by ignoring dilution; the Bayes factor rarely gave false alarms but often missed mild dilution. 
In a hepatitis~E serosurvey design on NHANES data, ignoring dilution nearly halved the estimated baseline odds, and BADGER was less biased than the logistic submodel.

%% file: sections/01-introduction.tex
\section{Introduction}
\label{sec:intro}

The idea behind group testing \citep{Dorfman1943} is simple: rather than testing every specimen on its own, several are mixed and tested together. If the pool is negative, all of its members are cleared with a single test; if it is positive, its members are retested individually. When infection is rare, this saves a great many tests, and pooling has become a routine part of screening. Blood services, for example, screen donations for hepatitis~E virus (HEV) RNA in minipools of between $6$ and $96$ \citep{Harvala2019,Boland2019}, and pooled testing has also been used for chlamydia and gonorrhoea \citep{Lewis2012}, HIV \citep{Kline1989,Krajden2014}, influenza \citep{Van2012} and, more recently, SARS-CoV-2 \citep{Barak2021}. Often the aim is not only to identify who is infected but also to understand who is likely to be. In \emph{group testing regression}, the probability of infection is linked to individual covariates, and this relationship is estimated from the pooled outcomes, together with any retests \citep{Farrington1992,Vansteelandt2000,Xie2001,Bilder2009,Chen2009,Delaigle2011,Wang2014,McMahan2017,ChatterjeeBandyopadhyay2020}.

Pooling has a cost. Mixing a positive specimen with negatives dilutes it, so whether a pool tests positive depends on its size $c$ and on the number $k$ of positives it contains. In SARS-CoV-2 testing, a pool of size $c$ with one positive has a cycle threshold about $\log_2c$ higher than the positive specimen alone \citep{Barak2021}, and estimated sensitivity fell from $93\%$ to $91\%$ and $81\%$ as pools grew from $5$ to $10$ and $50$ \citep{Bateman2021}. For HEV, individual testing found $17$ RNA-positive donations among $10{,}141$, whereas pools of $96$ found only $8$ \citep{Vollmer2019}, and pooling ten sera, though not five, missed low-titer antibody positives \citep{Emmanuel1988}. If pooled sensitivity is assumed constant, positives lost to dilution are mistaken for lower prevalence.

Dilution was first studied for estimating a single prevalence \citep{Hwang1976,ZeniosWein1998,HungSwallow1999}. In regression, one approach models the continuous biomarker behind each test and treats the pooled value as an average of individual values \citep{McMahan2013,Delaigle2015,Wang2015,Mokalled2021,Self2022}, with Bayesian versions for regression \citep{Self2022} and classification \citep{Tatsuoka2023}. These methods need a biomarker distribution and a threshold rule. When only binary outcomes are available, \citet{Warasi2017} model pooled sensitivity as a logistic, probit or complementary log-log function of $k/c$ with a single parameter $\lambda$, estimate it alongside the coefficients, and test $\lambda=0$ with a likelihood ratio test. \citet{Cho2026} test, without covariates, whether the false-negative rate of a pool rises with its size. In both cases the shape of the curve is assumed, and a wrong shape can distort both the curve and the coefficients.

We propose BADGER (\textbf{B}ayesian \textbf{A}nalysis of \textbf{D}ilution in \textbf{G}roup t\textbf{E}sting \textbf{R}egression), which avoids choosing a shape: the dilution curve can be any monotone function and is given a Dirichlet mixture prior, as in Bayesian monotone regression \citep{GelfandKuo1991,NeelonDunson2004,BornkampIckstadt2009}. Our contributions are as follows.
\begin{enumerate}
\item \emph{A shape-free curve shared across pool sizes.} Pooled sensitivity is a nondecreasing function $F$ of $k/c$, written as a mixture of distribution functions on $(0,1]$ with Dirichlet weights. The components are either step functions at the observed fractions, which give the saturated model, or a Bernstein basis. Because $F$ is shared, one pool size can inform fractions that another cannot reach.
\item \emph{An augmentation.} We give each pooled reading a latent \emph{cause}: a false positive, a curve component or an assay miss. Together with the individual statuses, drawn in $O(c_j^2)$ operations per pool, the causes yield Dirichlet and beta full conditionals for the curve, the pooled false-positive rate and the assay accuracy. We compute the Bayes factor for dilution by bridge sampling \citep{MengWong1996}.
\item \emph{Identification.} The pooled likelihood is unchanged if the intercept is shifted and the curve is re-weighted pool by pool to compensate. Even so, positive pools inform the slopes whatever the curve, and monotonicity confines the intercept to an explicit interval, within which the curve prior decides. This interval predicts the direction of the bias in any fit that models dilution.
\item \emph{Evidence.} In simulations with pools of sizes $5$ and $10$, we compare BADGER with the three submodels of \citet{Warasi2017} and the fit that ignores dilution, reporting paired Monte Carlo standard errors and both error rates of the Bayes factor. We also design a pooled HEV serosurvey on the 2017--2018 National Health and Nutrition Examination Survey (NHANES), using a mechanistic curve that none of the fitted models contains.
\end{enumerate}
Because no public data set combines pooled outcomes, retests and covariates, \cref{sec:app} follows \citet{McMahan2013}, \citet{Warasi2017}, \citet{Mokalled2021} and \citet{Self2022} in pooling individually tested specimens at random.

\Cref{sec:model} sets out the model and \cref{sec:bayes} describes BADGER. \Cref{sec:sim} reports the simulations and \cref{sec:app} the design study, and \cref{sec:discussion} concludes. Derivations and further results are in the Supplementary Material.

%% file: sections/02-model.tex
\section{Group Testing Regression with Dilution}
\label{sec:model}

\subsection{Individuals, Pools and the Dorfman Protocol}

Suppose $N$ individuals are split into $J$ non-overlapping pools, with $c_j$ members in pool $j$. Member $i$ of pool $j$ has true status $\Ytil_{ij}\in\{0,1\}$ and covariates $\bX_{ij}=(1,x_{ij1},\dots,x_{ijp})\tp$, and the statuses are independent with
\begin{equation*}
\Ytil_{ij}\mid\bX_{ij}\sim\mathrm{Ber}(p_{ij}),\qquad \mathrm{logit}(p_{ij})=\bX_{ij}\tp\bbeta .
\end{equation*}
We write $K_j=\sum_{i=1}^{c_j}\Ytil_{ij}$ for the number of positives in pool $j$. Under the protocol of \citet{Dorfman1943}, each pool is tested once, giving $Z_j\in\{0,1\}$. If the pool is negative, all its members are declared negative; if it is positive, every member is retested, giving independent
\begin{equation*}
Y_{ij}\mid\Ytil_{ij}\sim\mathrm{Ber}\{\Se\,\Ytil_{ij}+(1-\Sp)(1-\Ytil_{ij})\},
\end{equation*}
where $\Se$ and $\Sp$ are the sensitivity and specificity of an individual test. The observed data are $\mathcal D=\{\bZ,\bY,\bX\}$, where $\bY$ is available only for positive pools.

\subsection{Dilution Submodels}

A dilution submodel describes how the chance that a pool tests positive depends on what it contains:
\begin{equation*}
h(k,c)=\Prob(Z_j=1\mid K_j=k,\ c_j=c),\qquad k=0,1,\dots,c ,
\end{equation*}
Here $h(0,c)=\varphi$ is the pooled false-positive rate, $h(c,c)=\Se$, and $h(\cdot,c)$ is nondecreasing: adding positives to a pool never makes it less likely to test positive. \citet{Warasi2017} set $\varphi=1-\Sp$ and, writing $\tau=(k-c)/c$, consider three forms:
\begin{equation}
\begin{split}
\text{logistic:}\quad & h(k,c)=\frac{\exp(\lambda\tau)}{1/\Se+\exp(\lambda\tau)-1},\\
\text{probit:}\quad & h(k,c)=\Phi\{\Phi^{-1}(\Se)+\lambda\tau\},\\
\text{complementary log-log:}\quad & h(k,c)=1-\exp\bigl(-\exp[\log\{-\log(1-\Se)\}+\lambda\tau]\bigr),
\end{split}
\label{eq:param}
\end{equation}
for $k\ge1$, where $\lambda\ge0$ and $\lambda=0$ means no dilution. Because $\tau$ depends only on $k/c$, a single $\lambda$ covers every pool size. This holds if the pooled concentration is the average of the members' concentrations and positive specimens have similar concentrations. With a single pool size we write $h(k)$. We keep the dependence on $k/c$, but we let $\varphi$ differ from $1-\Sp$, since contamination or cross-reactivity summed over sera can raise it, and we treat $\Se$ and $\Sp$ as unknown.

\subsection{Likelihood}

For $\bu\in\{0,1\}^{c_j}$ let $P_j(\bu)=\prod_{i}p_{ij}^{u_i}(1-p_{ij})^{1-u_i}$, $\psi_j(\bu)=\{\prod_if(Y_{ij}\mid u_i)\}^{Z_j}$ with $f(y\mid1)=\Se^y(1-\Se)^{1-y}$ and $f(y\mid0)=(1-\Sp)^y\Sp^{1-y}$, and $g_j(k)=h(k,c_j)^{Z_j}\{1-h(k,c_j)\}^{1-Z_j}$. The likelihood factorizes over pools,
\begin{equation}
L(\btheta\mid\mathcal D)=\prod_{j=1}^{J}\ \sum_{\bu\in\{0,1\}^{c_j}}P_j(\bu)\,\psi_j(\bu)\,g_j(|\bu|),
\label{eq:lik}
\end{equation}
where $|\bu|=\sum_iu_i$ and $\btheta$ collects $\bbeta$, the curve and the accuracy parameters. Each term of the inner sum corresponds to one possible configuration $\bu$ of true statuses in pool $j$: 
\begin{itemize}
 \item $P_j(\bu)$ is the probability of that configuration under the regression model, 
 \item $\psi_j(\bu)$ is the probability of the retest results given $\bu$ (equal to $1$ for a negative pool, which is not retested), and 
 \item $g_j(|\bu|)$ is the probability of the observed pool result given its $|\bu|$ positives, which is where dilution enters. 
\end{itemize}
Summing over $\bu$ averages over the unobserved statuses. The inner sum has $2^{c_j}$ terms, but grouping them by the number of positives $|\bu|$ gives
\begin{equation}
\sum_{\bu}P_j(\bu)\,\psi_j(\bu)\,g_j(|\bu|)=\sum_{k=0}^{c_j}q_j(k)\,g_j(k),\qquad
q_j(k)=\sum_{|\bu|=k}P_j(\bu)\,\psi_j(\bu).
\label{eq:qjk}
\end{equation}
With $\alpha_{ij}=(1-p_{ij})f(Y_{ij}\mid0)^{Z_j}$ and $\gamma_{ij}=p_{ij}f(Y_{ij}\mid1)^{Z_j}$, the $q_j(k)$ follow from the Poisson--binomial recursion
\begin{equation}
q^{(i)}_j(k)=\alpha_{ij}\,q^{(i-1)}_j(k)+\gamma_{ij}\,q^{(i-1)}_j(k-1),\qquad q^{(0)}_j(k)=\mathbf 1\{k=0\},\quad q^{(i)}_j(-1)=0,
\label{eq:pbrec}
\end{equation}
run over $i=1,\dots,c_j$, so that $q_j(k)=q^{(c_j)}_j(k)$. This takes $O(c_j^2)$ operations rather than $O(2^{c_j})$. \citet{Warasi2017} maximize \cref{eq:lik} under \cref{eq:param} with an EM algorithm.

%% file: sections/03-bayes.tex
\section{BADGER}
\label{sec:bayes}

\subsection{A Shape-Free Dilution Curve Shared Across Pool Sizes}
\label{sec:curve}

We define the pooled sensitivity, with $k$ positives in a pool of size $c$, as
\begin{equation}
h(0,c)=\varphi,\qquad h(k,c)=\varphi+(\Se-\varphi)\,F(k/c),\quad k=1,\dots,c,
\label{eq:hF}
\end{equation}
where $\varphi > 0$ and $\Se > \varphi$ are the lower and upper bounds of $h$. The \emph{dilution curve} $F:(0,1]\to[0,1]$ is nondecreasing with $F(1)=1$ and $F\equiv1$ indicates no dilution. We take
\begin{equation}
F(x)=\sum_{m=1}^{L}w_m\,G_m(x),\qquad \bm w=(w_1,\dots,w_L)\in\Delta_L,
\label{eq:mixture}
\end{equation}
a mixture of $L$ distribution functions $G_m$ on $(0,1]$ with $G_m(1)=1$, where $\Delta_L$ is the probability simplex. Every $\bm w$ then gives an $h$ with the three properties of \cref{sec:model}. With step functions and Dirichlet weights this is the monotone-curve prior of \citet{GelfandKuo1991}; with beta distribution functions it gives the random Bernstein polynomials of \citet{Petrone1999}.

\paragraph{Fraction grid} Let $0<t_1<\dots<t_L=1$ be the distinct fractions $k/c_j$, $1\le k\le c_j$, in the data, and $G_m(x)=\mathbf 1\{x\ge t_m\}$. Then $F$ is an arbitrary nondecreasing step function on the grid with jump $w_m$ at $t_m$, and the model is \emph{saturated}: it contains every monotone curve at the observed fractions, including each submodel \cref{eq:param} whose values stay in $[\varphi,\Se]$. Since $F$ is evaluated only at the observed fractions, $\bm w^0=\bm e_1$, which jumps to $1$ at the smallest one, is no dilution. For one pool size $c$, $L=c$, $t_m=m/c$, and $(\Se-\varphi)w_r=h(r)-h(r-1)$ is the increase in pooled sensitivity from $r-1$ to $r$ positives.

\paragraph{Bernstein basis} With mixed or large pool sizes most grid points are reached only by pools with several positives, which are rare. We then take $G_1(x)=\mathbf 1\{x>0\}$ and, for $m=2,\dots,L$, $G_m$ the Beta$(m-1,L-m+1)$ distribution function, so that $F$ is a Bernstein polynomial of degree $L-1$ with nondecreasing coefficients plus a jump at $0^+$. It borrows strength between neighboring fractions; as for the grid, no dilution is $\bm w^0=\bm e_1$. \Cref{fig:bases} shows prior draws of $F$ under both bases.

\begin{figure}[!ht]
\centering
\IfFileExists{figures/dilution_bases.pdf}{\includegraphics[width=0.85\textwidth]{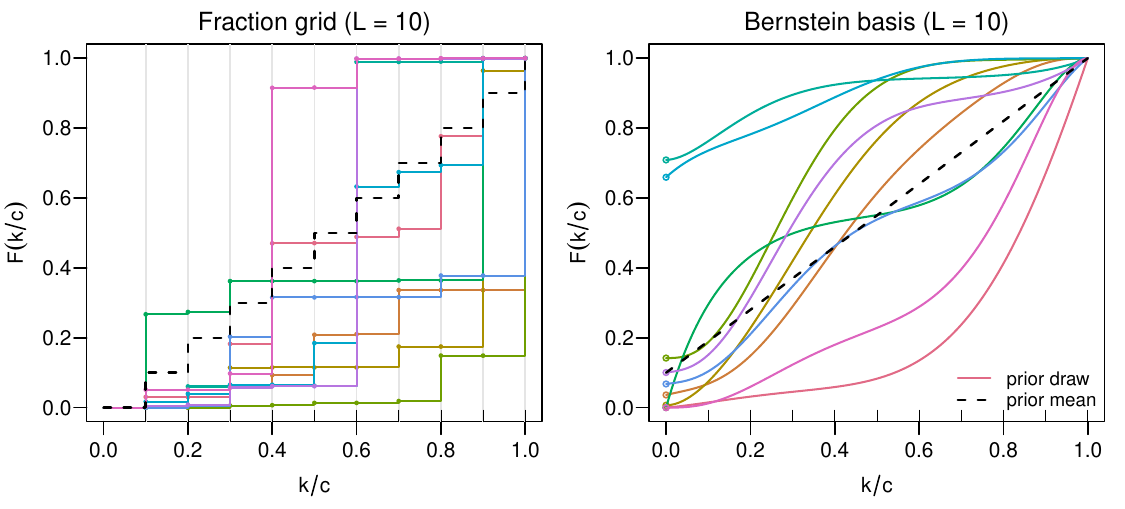}}{\fbox{generated by \texttt{figures/dilution\_bases.R}}}
\caption{Ten draws of the dilution curve $F$ from the prior, $\bm w\sim\mathrm{Dirichlet}(\bm 1_L/L)$, for pool size $10$ and $L=10$ in both bases. The dashed lines are the prior mean of $F$.}
\label{fig:bases}
\end{figure}

\subsection{Priors}
\label{sec:prior}

We take
\begin{equation}
\bm w\sim\mathrm{Dirichlet}(a\bm 1_L),\qquad \bbeta\sim N(\bm0,\sigma_\beta^2I),
\label{eq:prior}
\end{equation}
with default $a=1/L$ (prior P1), which for one pool size is the saturated model with $a=1/c$. The prior mean of $F$ is the average of the basis functions, $\{1+(L-1)x\}/L$ for the Bernstein basis and $x$ for the grid with one pool size. Since $\alpha < 1$, P1 is sparse, favoring each of the $L$ vertices equally, and only one vertex is no dilution, so it favors deep curves. For example, with $\Se=\Sp=0.99$ and pool sizes $5$ and $10$, the prior mean of $h(1,10)$ is $0.11$ (grid) and $0.23$ (Bernstein). Another choice of prior (P2) is $(a_1,\dots,a_L)=(0.5,0.5/(L-1),\dots,0.5/(L-1))$, under which $\E\{F(0^+)\}=0.5$, the prior median of $h(1,10)$ is $0.56$ and $\Prob\{h(1,10)>0.9\}=0.21$. Instead of a prior with mass at the no-dilution curve, we compare the two models by a Bayes factor (\cref{sec:bf}).

We write $\varphi=\Se\,\rho$ with $\rho\sim\mathrm{Beta}(a_\rho,b_\rho)$, and take $\Se\sim\mathrm{Beta}(a_{\Se},b_{\Se})$ and $\Sp\sim\mathrm{Beta}(a_{\Sp},b_{\Sp})$. With known accuracy we set $\varphi=1-\Sp$, as \citet{Warasi2017} do. The product form keeps the conditionals conjugate, and since $\Se\approx1$, $\rho\approx\varphi$, so the prior on $\rho$ is in effect a prior on $\varphi$.

\subsection{Augmentation and Sampler}
\label{sec:gibbs}

The posterior has no standard form: $g_j(k)$ is affine in $\bm w$, and each pool contributes a mixture of $2^{c_j}$ logistic likelihoods in $\bbeta$. We therefore augment the data with the true statuses $\widetilde{\bY}$ and with a latent \emph{cause} $R_j$ for each pooled reading, which records what produced it. 
The augmented variable $R_j$ takes values in $\{1,\dots,L\}$, the components of the curve, plus "fp" for a false positive and "miss" for an assay miss, and its distribution is defined by
\begin{align*}
    \Prob(R_j=\mathrm{fp} \mid K_j = k) & =\varphi = \Se\rho,\quad \Prob(Z_j = 1\mid R_j = \mathrm{fp}) = 1,\\ 
    \Prob(R_j=m \mid K_j = k) & = (\Se - \varphi)w_m,\quad \Prob(Z_j = 1\mid R_j = m, K_j=k) = G_m(k/c_j), \quad m=1,\dots,L,\\
    \Prob(R_j = \mathrm{miss} \mid K_j = k) & = 1 - \Se,\quad \Prob(Z_j = 1\mid R_j = \mathrm{miss}) = 0,
\end{align*}
See \Cref{fig:causes} for an illustration. Because $G_m(0)=0$, averaging over $R_j$ recovers the dilution model:
\begin{align*}
\Prob(Z_j=1\mid K_j=k) = \sum_{r} \Prob(Z_j=1\mid R_j=r,K_j=k) \Prob(R_j=r\mid K_j=k)\\
 =\varphi+(\Se-\varphi)\sum_{m = 1}^L w_mG_m(k/c_j)=h(k,c_j),\qquad k=0,1,\dots,c_j.
\end{align*}
When $K_j=0$, every component reads negative and the data cannot tell them apart, so we record only that the cause was some component, which has probability $\Se-\varphi$, and not which one.

\begin{figure}[!t]
\centering
\IfFileExists{figures/cause_tree.pdf}{\includegraphics[width=0.85\textwidth]{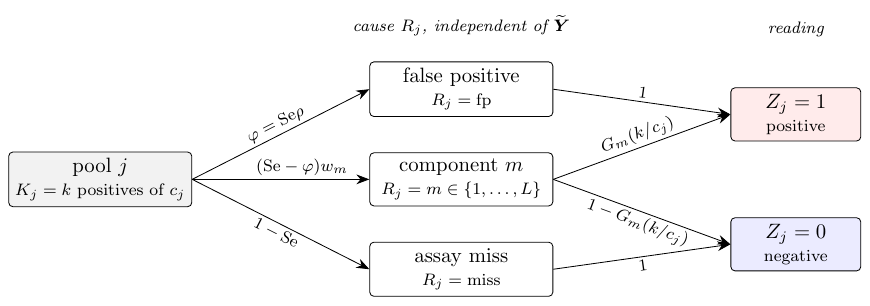}}{\fbox{generated by \texttt{figures/cause\_tree.tex}}}
\caption{The latent cause $R_j$ of the reading of pool $j$, with $K_j=k$ positives among its $c_j$ members.}
\label{fig:causes}
\end{figure}

Let $n_{\mathrm{fp}}$, $n_{\mathrm{comp}}$ and $n_{\mathrm{miss}}$ count the pools with $R_j=\mathrm{fp}$, $R_j\in\{1,\dots,L\}$ and $R_j=\mathrm{miss}$ (for a pool with $K_j=0$ every component reads negative, so the sampler records only that the cause was a component, not which one), $n_m$ the pools with $K_j\ge1$ and $R_j=m$, and $d_{yz}$ the retested members with status $y$ and reading $z$ (\cref{fig:counts}). The complete-data likelihood is then
\begin{equation*}
\Se^{\,n_{\mathrm{fp}}+n_{\mathrm{comp}}+d_{11}}(1-\Se)^{n_{\mathrm{miss}}+d_{10}}\;\rho^{n_{\mathrm{fp}}}(1-\rho)^{n_{\mathrm{comp}}}\;\prod_m w_m^{n_m}\;\Sp^{d_{00}}(1-\Sp)^{d_{01}}\;\prod_{ij}p_{ij}^{\Ytil_{ij}}(1-p_{ij})^{1-\Ytil_{ij}},
\end{equation*}
times factors $G_m(K_j/c_j)^{Z_j}\{1-G_m(K_j/c_j)\}^{1-Z_j}$ for pools with $R_j=m$, which do not involve the parameters. 

\begin{figure}[!ht]
\centering
\IfFileExists{figures/count_tree.pdf}{\includegraphics[width=0.70\textwidth]{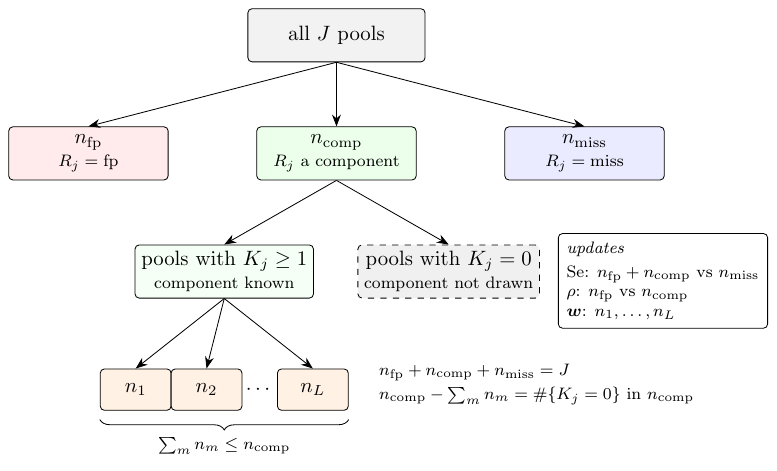}}{\fbox{generated by \texttt{figures/count\_tree.tex}}}
\caption{How the cause counts partition the $J$ pools. Each pool has one cause, so $n_{\mathrm{fp}}$, $n_{\mathrm{comp}}$ and $n_{\mathrm{miss}}$ add up to $J$. Pools caused by a component split into those with $K_j\ge1$, counted by component in $n_1,\dots,n_L$, and those with $K_j=0$, whose component is not drawn. The box lists the counts used by each conjugate update in \cref{eq:conj}.}
\label{fig:counts}
\end{figure}

The curve weights, $\rho$ and the accuracy parameters therefore have Dirichlet or beta full conditionals, and the sampler cycles through
\begin{subequations}
\label{eq:conj}
\begin{align}
\Prob(K_j=k\mid\cdot)&\ \propto\ q_j(k)\,g_j(k),\qquad k=0,\dots,c_j, \label{eq:condK}\\
\Prob(\widetilde{\bY}_j=\bu\mid K_j=k,\cdot)&\ =\ \frac{P_j(\bu)\,\psi_j(\bu)}{q_j(k)},\qquad |\bu|=k, \label{eq:condY}\\
\Prob(R_j=r\mid K_j,Z_j,\cdot)&\ \propto\ \Prob(R_j=r)\,\Prob(Z_j\mid R_j=r,K_j), \label{eq:condR}\\
\bm w\mid\cdot&\ \sim\ \mathrm{Dirichlet}(a\bm1+\bm n), \label{eq:condw}\\
\rho\mid\cdot&\ \sim\ \mathrm{Beta}(a_\rho+n_{\mathrm{fp}},\ b_\rho+n_{\mathrm{comp}}),\\
\Se\mid\cdot&\ \sim\ \mathrm{Beta}(a_{\Se}+n_{\mathrm{fp}}+n_{\mathrm{comp}}+d_{11},\ b_{\Se}+n_{\mathrm{miss}}+d_{10}),\\
\Sp\mid\cdot&\ \sim\ \mathrm{Beta}(a_{\Sp}+d_{00},\ b_{\Sp}+d_{01}),\\
p(\bbeta\mid\cdot)&\ \propto\ \pi(\bbeta)\prod_{i,j}p_{ij}^{\Ytil_{ij}}(1-p_{ij})^{1-\Ytil_{ij}}, \label{eq:condbeta}
\end{align}
\end{subequations}
where $\bm n=(n_1,\dots,n_L)$ and ``$\cdot$'' denotes all other quantities. We draw $K_j$ from \cref{eq:condK} using \cref{eq:pbrec} and then $\widetilde{\bY}_j$ from \cref{eq:condY} without enumerating its $\binom{c_j}{k}$ values, working backwards through the members, $i=c_j,\dots,1$: set $\Ytil_{ij}=1$ with probability $\gamma_{ij}\,q^{(i-1)}_j(k-1)/q^{(i)}_j(k)$ and reduce $k$ by $\Ytil_{ij}$ (\cref{fig:backward}). In \cref{eq:condR} only the causes consistent with $Z_j$ have positive probability (\cref{fig:causes}), and when $K_j=0$ the component is not drawn. Given $\widetilde{\bY}$, \cref{eq:condbeta} is the posterior of an ordinary logistic regression on $\bX$; we update $\bbeta$ by an independence Metropolis step whose proposal is the normal approximation at the complete-data maximum likelihood estimate, although P\'olya--Gamma augmentation \citep{PolsonScottWindle2013} would give an exact draw. Because \cref{eq:condK} sums out $R_j$, $R_j$ must be drawn after $\widetilde{\bY}_j$, not before, which makes the sweep a partially collapsed block Gibbs sampler \citep{vanDykPark2008}. The causes of pools with $K_j=0$ enter no conditional when $\rho$ and $\Se$ are fixed, and drawing them would only slow the mixing of $\bm w$; in our checks, skipping them roughly doubled the effective sample size of the curve without changing the posterior.

\begin{figure}[!ht]
\centering
\IfFileExists{figures/backward_draw.pdf}{\includegraphics[width=0.85\textwidth]{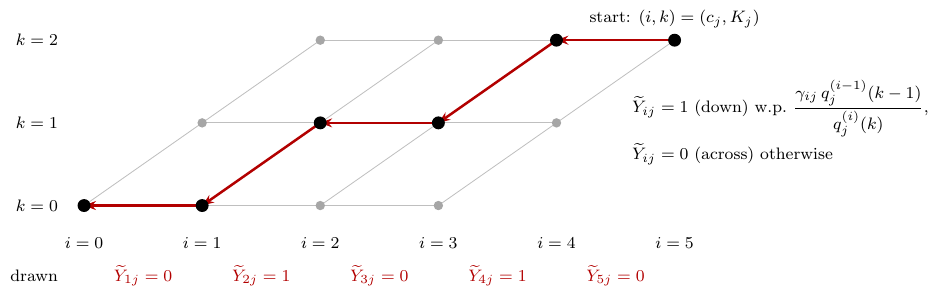}}{\fbox{generated by \texttt{figures/backward\_draw.tex}}}
\caption{Drawing the statuses of a pool given $K_j$, for $c_j=5$ and $K_j=2$. A state $(i,k)$ means that $k$ positives remain to be placed among members $1,\dots,i$; only states that can reach $(0,0)$ are shown. From $(i,k)$ the draw moves down ($\Ytil_{ij}=1$) with the probability shown, or across ($\Ytil_{ij}=0$). The highlighted path is one draw, and it ends at $(0,0)$ with exactly $K_j$ positives.}
\label{fig:backward}
\end{figure}

\subsection{Bayes Factor for a Dilution Effect}
\label{sec:bf}

A likelihood ratio test for dilution is awkward since no dilution is a vertex of $\Delta_L$, so the null distribution is a chi-bar-square mixture with weights that depend on the unknown information \citep{Kudo1963,SilvapulleSen2005}, or $\tfrac12\chi^2_0+\tfrac12\chi^2_1$ in a submodel \citep{SelfLiang1987}. A Bayes factor needs no null distribution. Let $\mathcal M_0$ fix $\bm w=\bm w^0$ and $\mathcal M_1$ take $\bm w\sim\mathrm{Dirichlet}(a\bm1)$. With known accuracy and $\varphi=1-\Sp$, $B_{10}=m_1/m_0$, where
\begin{equation*}
m_1(\mathcal D)=\int I(\bm w)\,\mathrm{Dir}(\bm w;a\bm1)\,d\bm w,\qquad m_0(\mathcal D)=I(\bm w^0),\qquad I(\bm w)=\int L(\bbeta,\bm w\mid\mathcal D)\,\pi(\bbeta)\,d\bbeta .
\end{equation*}
We estimate $m_0$ by importance sampling and $m_1$ by bridge sampling \citep{MengWong1996}, with a proposal that follows the ridge of \cref{sec:ident}. It agreed with an independent importance-sampling estimate to within $0.08$ on the $\log_{10}$ scale (Supplement~\sref{sec:bfcheck}). 

The Bayes factor depends on the prior under $\mathcal M_1$. P1 puts most of its mass on severe dilution, so mildly diluted data can favor $\mathcal M_0$: a small $B_{10}$ is evidence against the curves the prior expects, not evidence for no dilution. \Cref{sec:sim} therefore reports both kinds of misleading evidence \citep{SchonbrodtWagenmakers2018} under P1 and P2.

\subsection{Identifiability of the Intercept and the Curve}
\label{sec:ident}

In this subsection, we study the identifiability of the intercept and the curve. Throughout, we assume $\Se$ and $\Sp$ are known. Recall that $\text{logit}(p_{ij}) =\bX_{ij}\tp\bbeta$ and let $\bm e_1$ be the first coordinate vector. The moment generating function of $K_j$ is
\begin{equation*}
\begin{aligned}
A_j(\delta) &= \sum_{k=0}^{c_j}\Prob(K_j=k)\,e^{k\delta}
             = \sum_{\bu\in\{0,1\}^{c_j}} P_j(\bu)\,e^{|\bu|\delta}\\
            &= \sum_{\bu\in\{0,1\}^{c_j}} \prod_{i=1}^{c_j} (1-p_{ij})^{1-u_i}(p_{ij}e^{\delta})^{u_i}
             = \prod_{i=1}^{c_j}\{1+(e^\delta-1)p_{ij}\}, \qquad \delta\in\R.
\end{aligned}
\end{equation*}
Shifting the intercept by $\delta$ multiplies each member's odds by $e^{\delta}$, so the status probabilities of pool $j$ become
\begin{equation*}
P^\delta_j(\bu)=\frac{P_j(\bu)\,e^{|\bu|\delta}}{A_j(\delta)},\qquad \bu\in\{0,1\}^{c_j}.
\end{equation*}
Let $h_j^\delta(k)=h(k,c_j)\,e^{-k\delta}A_j(\delta)$. Then 
\[
P^\delta_j(\bu)\,\psi_j(\bu)\,h^\delta_j(|\bu|)=\frac{P_j(\bu)e^{|\bu|\delta}}{A_j(\delta)}\,\psi_j(\bu)\,h(|\bu|,c_j)e^{-|\bu|\delta}A_j(\delta)=P_j(\bu)\,\psi_j(\bu)\,h(|\bu|,c_j),
\] 
i.e., the $j$th factor of \cref{eq:lik} takes the same value at $(\bbeta,h)$ as at $(\bbeta+\delta\bm e_1,h_j^\delta)$. In other words, a higher prevalence paired with a correspondingly lower curve fits every pool as well as the original: pooled readings alone cannot tell more positives that are diluted more from fewer positives that are diluted less. The intercept is identified only insofar as $h^\delta_j$ fails to be an admissible curve, that is, one that is nondecreasing, runs from $\varphi$ at $k=0$ to $\Se$ at $k=c_j$, and is shared by all pools.

This identifiability issue can be resolved by considering the data from all pool sizes together and this is why BADGER shares one curve across pool sizes. For simplicity, we consider the case of no covariates. The case with covariates follows the same principle. With no covariates, write $p_{ij}\equiv p$ and furthermore assume that $\Se+\Sp\ne1$. Then $P_j(\bu)=p^{|\bu|}(1-p)^{c_j-|\bu|}$ depends on $\bu$ only through $k=|\bu|$, and grouping the terms of \cref{eq:lik} by $k$ gives
\begin{equation*}
L(p,h\mid\mathcal D)=\prod_{j=1}^J\sum_{k=0}^{c_j}p^k(1-p)^{c_j-k}\,S_j(k)\,h(k,c_j)^{Z_j}\{1-h(k,c_j)\}^{1-Z_j},\qquad S_j(k)=\sum_{|\bu|=k}\psi_j(\bu),
\end{equation*}
where $S_j(k)=\binom{c_j}{k}$ for a negative pool, whose members are not retested, and $S_j(k)$ carries the retests of a positive pool.
\begin{proposition}[Shared fractions]
\label{prop:shared}
Suppose there are no covariates, that $\Se+\Sp\ne1$, and that the data contain pools of two sizes $c<c'$ that share a fraction $x=k/c=k'/c'$ with $0<x<1$. If $(p,h)$ and $(p',h')$, with $p,p'\in(0,x)$ and $h(k,c)>0$, satisfy $L(p,h\mid\mathcal D)=L(p',h'\mid\mathcal D)$ for every data set $\mathcal D$, then $p=p'$ and $h(k,c)=h'(k,c)$ for every pool size $c$ in the data and $k=0,\dots,c$.
\end{proposition}

The shared curve takes the same unknown value at a fraction common to both sizes, so comparing how often the two cells occur cancels the curve and leaves a function of the prevalence alone. This is information that no single pool size can provide, and it costs little: with sizes $5$ and $10$ the smallest shared fraction, $0.2$, is far above the usual prevalence. The proofs are in Supplement~\sref{sec:identproof}. 

%% file: sections/04-simulation.tex
\section{Simulation Studies}
\label{sec:sim}

\subsection{Simulation Design}
\label{sec:simdesign}

Each data set has $N=3000$ individuals in $300$ pools of five and $150$ pools of ten, tested by Dorfman's procedure with $\Se=\Sp=0.99$. Statuses follow $\mathrm{logit}\{\Prob(\Ytil_{ij}=1)\}=\beta_0+\beta_1x_{ij1}+\beta_2x_{ij2}$ with $\bbeta=(-3,2,1)$, $x_{ij1}\sim N(0,0.5^2)$ and $x_{ij2}\sim\mathrm{Ber}(0.1)$, so about $8\%$ of individuals are positive. The false-positive rate is $\varphi=1-\Sp$, and the true curve, shared by both sizes, is
\begin{equation}
h(k,c)=(1-\Sp)+z+(C-z)\,v(k/c),\qquad C=\Se-(1-\Sp),
\label{eq:truecurve}
\end{equation}
where $v$ rises from $v(0)=0$ to $v(1)=1$ and $z\in\{0.9,0.7,0.5\}$ sets mild, moderate or severe dilution. We take $v$ concave ($x^{1/2}$), linear, convex ($x^2$), concave--convex or convex--concave (halves joined at $x=1/2$). Two further shapes come from \citet{Warasi2017}: the logistic submodel \cref{eq:param} with $\lambda\in\{2.6,3.8,5\}$, and the Hung--Swallow curve $h(k,c)=\Se\,x/\{x+\lambda(1-x)\}$ \citep{HungSwallow1999}, with $\lambda\in\{0.05,0.10,0.15\}$, which drops steeply at small fractions ($h(1,10)=0.68,0.52,0.42$). With no dilution ($h\equiv\Se$) this gives $22$ settings, each replicated $200$ times.

We fit BADGER (Bernstein, $L=7$, P1) and three variants: BADGER-P2, BADGER-grid ($L=10$) and BADGER-sep, which has a separate grid curve for each size and is compared with the shared curves in \cref{tab:coef}. The comparators are the maximum likelihood fits of the three submodels \cref{eq:param} (LOG, PB, CLL) and GT$_0$, which ignores dilution. Every BADGER fit uses $\sigma_\beta=10$ and two chains of $12{,}000$ iterations, discarding $3000$ and thinning by $6$, which leaves $3000$ draws.

The curve error of an estimate $\widehat h$ in one data set is the weighted mean absolute error over the cells $(k,c)$, $1\le k<c$,
\begin{equation}
\mathrm{err}(\widehat h)=\frac{\sum_{c\in\{5,10\}}\sum_{k=1}^{c-1}\omega_{kc}\,\bigl|\widehat h(k,c)-h(k,c)\bigr|}{\sum_{c\in\{5,10\}}\sum_{k=1}^{c-1}\omega_{kc}},\qquad
\omega_{kc}=J_c\binom{c}{k}\bar p^{\,k}(1-\bar p)^{c-k},
\label{eq:curveerr}
\end{equation}
where $J_5=300$ and $J_{10}=150$ are the numbers of pools, $\bar p=0.078$ is the marginal prevalence, so $\omega_{kc}$ is the expected number of pools in cell $(k,c)$, and $\widehat h$ is the posterior mean for BADGER, the plug-in estimate for LOG, PB and CLL, and $\Se$ for GT$_0$; we report its mean over replications. The cells $(1,5)$ and $(1,10)$ carry about three quarters of the weight, and every cell with $k/c\ge0.5$ expects fewer than two pools (Supplement~\sref{sec:special}); unweighted errors are in Supplement~\sref{sec:mixedfull}. Monte Carlo standard errors (in parentheses) are paired, and we call two methods different only when their paired difference exceeds twice its standard error. The median effective sample size was $480$ to $1660$ for $\beta_0$ and $125$ to $670$ for $h(1,10)$, of $3000$ draws, and the largest split-$\widehat R$ over all replications of a setting was $1.14$ (Supplementary Table~\sref{tab:mixeddiag}).

\subsection{Estimation of the Coefficients}
\label{sec:simcoef}

\Cref{tab:coef} summarizes the estimates of $\bbeta$. Ignoring dilution biased $\widehat\beta_0$ by $-0.09$ at the mild level and $-0.52$ at the severe level on average, and every model of dilution removed most of this bias; with moderate or severe dilution BADGER had the smallest root mean squared error, although the models of dilution differed by at most $0.02$. Under no dilution they overestimated $\beta_0$, BADGER by $0.07$ $(0.01)$ and the parametric fits by $0.04$, as the end of the ridge in \cref{sec:ident} predicts; P2 lowered BADGER's bias by $0.02$ there and at the mild level. BADGER's $95\%$ intervals covered $\beta_0$ in $92\%$ to $99\%$ of replications and $\beta_1$ in $92\%$ to $98\%$, and every method recovered the slopes (Supplementary Tables~\sref{tab:mixedbeta} and~\sref{tab:mixedcov}). Two further studies (Supplement~\sref{sec:accfull}) relaxed the assumptions on the assay. When pools of negatives read positive with $\varphi=0.03$ against $1-\Sp=0.01$, learning $\varphi$ recovered it and removed or nearly halved the bias of $\widehat\beta_0$. When $\Se=0.95$ and $\Sp=0.98$ were learned under beta priors, the data said little about $\Se$, but every such fit covered $\beta_0$ in $94\%$ to $98\%$ of replications, against $90\%$ to $94\%$ with wrong fixed values.

\begin{table}
\caption{Estimates of the coefficients, mixed design: mean bias and root mean squared error of $\widehat\beta_0$, coverage of $95\%$ credible intervals for $\beta_0$ (Bayesian fits only), and root mean squared error of $\widehat\beta_1$ over all $22$ settings. Each entry is the mean over the seven shapes at that level; $200$ replications per setting. The truth is $\beta_0=-3$.}
\label{tab:coef}
\centering\small
\input{tables/mixed_coef}
\end{table}

\subsection{Bayes Factor for Testing Dilution}
\label{sec:simbf}

For testing the presence of dilution, a large $B_{10}$ is evidence for dilution and a small one evidence against it. We follow Jeffreys' scale as given by \citet{KassRaftery1995}, calling $B_{10}>10$ strong evidence for dilution and $B_{10}<1/10$ strong evidence for none. \Cref{tab:bfsim} shows that the Bayes factor never raised a false alarm and detected dilution more often as it became more severe, most easily for the Hung--Swallow curve. Mild dilution, however, was hard to tell from none: under P1 more than half of the mild replications gave strong evidence for no dilution. P2 never gave strong evidence for no dilution and detected dilution slightly more often, but outside the Hung--Swallow curve it also rarely found mild dilution (Supplement~\sref{sec:priorsens}).

\begin{table}[!h]
\caption{Frequency of strong evidence from the Bayes factor, mixed design: $B_{10}>10$ and $B_{10}<1/10$ under priors P1 and P2. Ranges are over the six shapes other than Hung--Swallow (HS); $200$ replications per setting (Supplementary Table~\sref{tab:mixedbf}).}
\label{tab:bfsim}
\centering\small
\input{tables/mixed_bfsum}
\end{table}

\subsection{Estimation of the Dilution Curve}
\label{sec:simcurve}

In the $21$ settings with dilution, BADGER had a smaller curve error than the best parametric fit in $11$, tied in $6$ and lost in $4$ (\cref{fig:paired}, which compares only BADGER and BADGER-grid with the parametric fits). \Cref{fig:mixed} shows that the fits differ mainly where few pools are expected. Supplement~\sref{sec:mixedfull} and~\sref{sec:priorsens} show how the basis (Bernstein or grid) and the prior (P1 or P2) affect the curve error. We recommend the grid for a single moderate pool size and the Bernstein basis with $L\approx7$ otherwise.

\begin{figure}[!h]
\centering
\IfFileExists{figures/mixed_paired.pdf}{\includegraphics[width=0.75\textwidth]{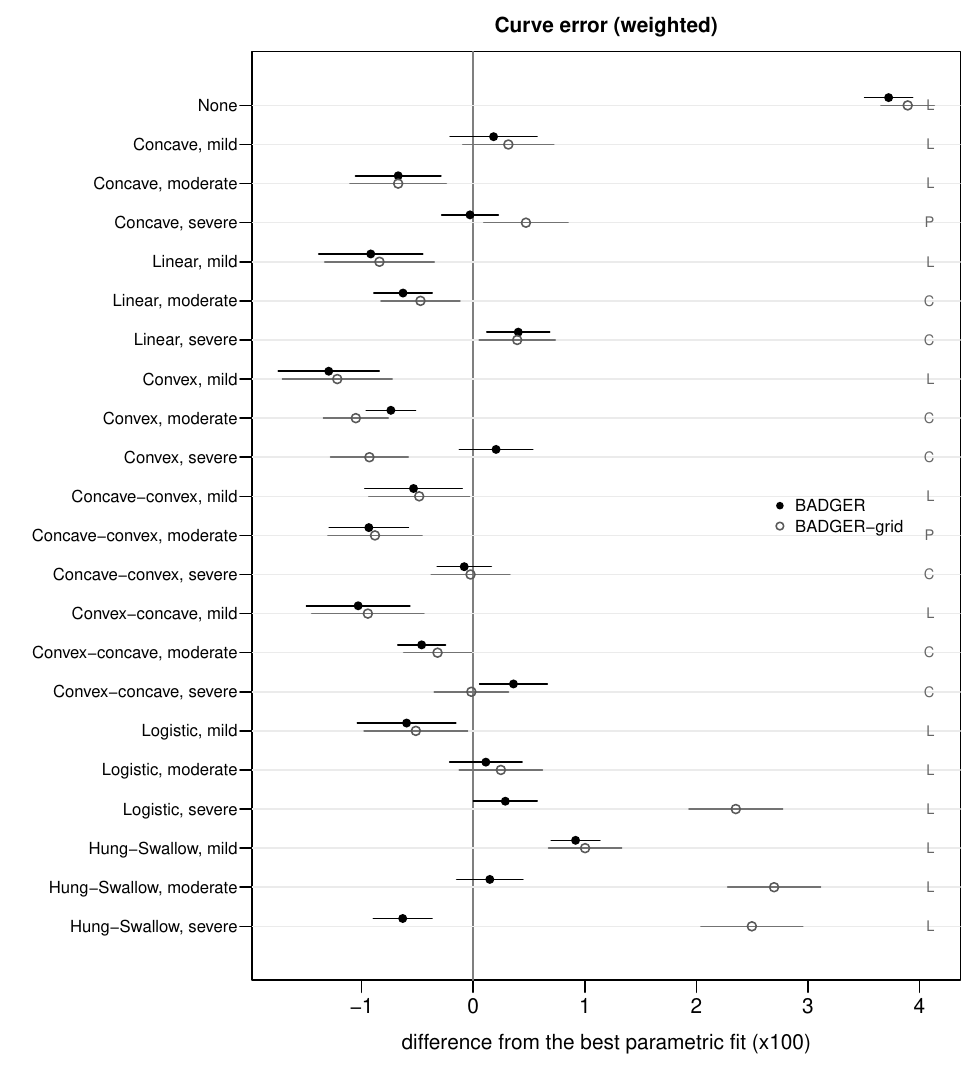}}{\fbox{generated by \texttt{summarise\_mixed.R}}}
\caption{Mixed pool sizes, $N=3000$, $200$ replications per setting: mean paired difference in weighted curve error ($\times100$) between BADGER (solid) or BADGER-grid (open) and the parametric fit with the smallest mean error in that setting (letter at the right: L, logistic; P, probit; C, complementary log-log), with $\pm2$ Monte Carlo standard errors; negative values favor the shape-free fit.}
\label{fig:paired}
\end{figure}

\begin{figure}[!h]
\centering
\IfFileExists{figures/mixed_curves.pdf}{\includegraphics[width=\textwidth]{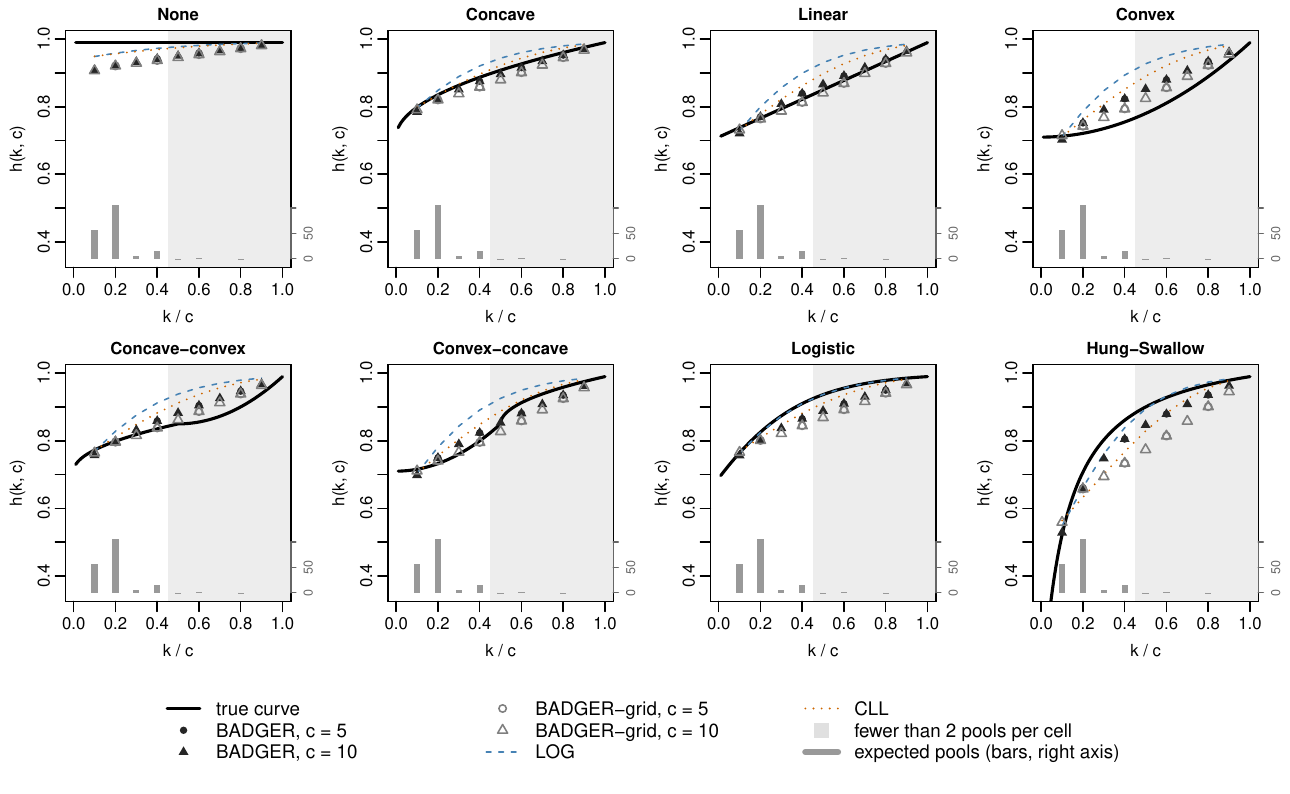}}{\fbox{generated by \texttt{summarise\_mixed.R}}}
\caption{Mean fitted curves over $200$ replications at the moderate level and under no dilution, against the fraction of positives $k/c$: BADGER (solid symbols) and BADGER-grid (open symbols) for pools of size $5$ (circles) and $10$ (triangles), the LOG and CLL fits (lines) and the true curve (thick line). Gray bars: expected number of pools at each fraction (right-hand axis); shaded region: fractions where fewer than two pools are expected in any cell.}
\label{fig:mixed}
\end{figure}

%% file: tables/mixed_coef.tex
\begin{tabular}{lccccccc}
\toprule
Dilution & GT$_0$ & LOG & CLL & BADGER & BADGER-grid & BADGER-sep & BADGER, P2 \\
\midrule
\multicolumn{8}{l}{\emph{Mean bias of $\widehat\beta_0$}} \\
None & \phantom{+}0.00 & +0.04 & +0.04 & +0.07 & +0.08 & +0.10 & +0.05 \\
Mild & $-$0.09 & +0.01 & +0.02 & +0.04 & +0.04 & +0.06 & +0.02 \\
Moderate & $-$0.27 & $-$0.02 & +0.01 & +0.01 & +0.02 & +0.03 & $-$0.01 \\
Severe & $-$0.52 & $-$0.07 & $-$0.02 & $-$0.04 & \phantom{+}0.00 & $-$0.01 & $-$0.04 \\
\midrule
\multicolumn{8}{l}{\emph{Root mean squared error of $\widehat\beta_0$}} \\
None & 0.100 & 0.113 & 0.114 & 0.125 & 0.126 & 0.142 & 0.114 \\
Mild & 0.152 & 0.132 & 0.136 & 0.131 & 0.132 & 0.137 & 0.127 \\
Moderate & 0.303 & 0.153 & 0.158 & 0.147 & 0.154 & 0.147 & 0.154 \\
Severe & 0.543 & 0.184 & 0.180 & 0.174 & 0.186 & 0.180 & 0.189 \\
\midrule
\multicolumn{8}{l}{\emph{Coverage of $95\%$ intervals for $\beta_0$}} \\
None & --- & --- & --- & 0.92 & 0.93 & 0.90 & 0.95 \\
Mild & --- & --- & --- & 0.96 & 0.95 & 0.94 & 0.96 \\
Moderate & --- & --- & --- & 0.96 & 0.96 & 0.96 & 0.96 \\
Severe & --- & --- & --- & 0.97 & 0.97 & 0.97 & 0.96 \\
\midrule
RMSE of $\widehat\beta_1$, all settings & 0.196 & 0.197 & 0.199 & 0.201 & 0.203 & 0.203 & 0.201 \\
\bottomrule
\end{tabular}

%% file: tables/mixed_bfsum.tex
\begin{tabular}{lcccc}
\toprule
 & \multicolumn{2}{c}{$B_{10}>10$} & \multicolumn{2}{c}{$B_{10}<1/10$} \\
\cmidrule(lr){2-3}\cmidrule(lr){4-5}
Dilution & P1 & P2 & P1 & P2 \\
\midrule
None & 0.00 & 0.00 & 0.86 & 0.00 \\
Mild & 0.01--0.02 (HS 0.34) & 0.03 (HS 0.46) & 0.55--0.72 & 0.00 \\
Moderate & 0.10--0.33 (HS 0.86) & 0.16--0.44 (HS 0.92) & 0.04--0.19 & 0.00 \\
Severe & 0.47--0.86 (HS 0.96) & 0.60--0.92 (HS 0.98) & 0.00--0.03 & 0.00 \\
\bottomrule
\end{tabular}

%% file: sections/05-application.tex
\section{Design Study: A Pooled Hepatitis~E Serosurvey}
\label{sec:app}

A serosurvey that pools sera risks missing low-titer specimens and must analyze pools of different sizes with unknown dilution. We study this with the hepatitis~E data of the 2017--2018 NHANES \citep{NHANES2018}.

\paragraph{Data and design} The response is anti-HEV IgG (\texttt{LBDHEG}), with covariates sex (\texttt{RIAGENDR}, $1$ for female) and standardized age (\texttt{RIDAGEYR}). We use $N=6790$ complete cases (prevalence $7.9\%$), ignoring the survey weights. The individual-data logistic regression gives $\widehat{\bbeta}=(-2.92,-0.18,1.21)$. Each of $100$ assignments per setting draws $6785$ participants, puts $3395$ in $679$ pools of size $5$ and $3390$ in $339$ pools of size $10$, and generates Dorfman readings, since NHANES tested every specimen individually. Each fit is compared with the individual-data estimate (IND) in the same assignment.

\paragraph{Assays and dilution} We consider two array accuracy settings: (i) the assay is accurate ($\Se=\Sp=0.99$) and (ii) the VIRCLIA IgG assay setting, $\Se=1.00$ and $\Sp=0.887$ \citep{MunozChimeno2024}, treated as known. The first mechanism is the logistic submodel with $\lambda=3.8$, written as $h=\varphi+(\Se-\varphi)F$ with $F$ taken at $\Se=\Sp=0.99$ and held fixed across assays, under which LOG submodel is correct. In the parametrization of \citet{Warasi2017} the sensitivity enters the curve, which would then show no dilution when $\Se=1$. The second is a \emph{threshold} curve \citep[cf.\ model DE1 of][]{HungSwallow1999}: positive sera carry titers with $\log T_i\sim N(0,1)$, and
\[
h(k,c)=\varphi+(\Se-\varphi)\,\frac{\Prob(\sum_{i\le k}T_i/c>\theta)}{\Prob(\sum_{i\le c}T_i/c>\theta)},
\]
with $\theta$ set so that $F(1/5)=0.6$. It is in no parametric family and not a function of $k/c$ alone ($F=0.73$ at $(2,10)$ against $0.60$ at $(1,5)$), so of the fitted models only BADGER-sep contains it. The third generates the threshold data at the titer level and makes a positive retest negative exactly when its titer is below the $1-\Se$ quantile, so that pooled and retest misses are correlated; it is used with the accurate assay only.

\paragraph{Methods} We fit BADGER (Bernstein, under P1 and P2), BADGER-grid, BADGER-sep, LOG and GT$_0$ as in \cref{sec:sim}.

\paragraph{Results} \Cref{tab:app} reports the intercept and the curve error \cref{eq:curveerr}, here with $J_5=679$, $J_{10}=339$ and $\bar p=0.079$. GT$_0$ underestimated $\beta_0$ by $0.21$ under the logistic mechanism and by $0.55$ to $0.67$ under the threshold mechanism, understating the baseline odds by a factor of up to $1.96$. Every model of dilution removed most of this. Under the logistic mechanism BADGER was close to LOG, with intercepts within $0.011$ $(0.002)$ and curve errors of $4.4$ against $4.0$ (accurate assay) and equal (VIRCLIA), so not assuming the shape cost little.

Under the threshold mechanism every model of dilution overestimated $\beta_0$, BADGER by $0.12$ to $0.14$ and LOG by $0.16$ to $0.17$ (paired standard errors $0.004$), and BADGER's curve error was lower by $1.3$ to $1.9$. The bias comes mainly from the data, not the prior: LOG, with no prior, was biased more. The threshold curve is within $0.06$ of $\Se$ for $k/c\ge0.4$, so a monotone fit can err only downward there, and by \cref{sec:ident} a curve that is too low is compensated by a higher intercept.

\Cref{tab:appbf} gives the coverage and the Bayes factor. Under the logistic mechanism the intervals almost always covered IND, but the Bayes factor often missed the dilution: under P1 it gave strong evidence for dilution in only $44\%$ and $15\%$ of assignments (accurate, VIRCLIA). Under the threshold mechanism it detected dilution every time, but the bias of $\widehat\beta_0$ pulled coverage down to $87\%$--$95\%$. The sex and age effects were within $0.06$ of IND under every model, GT$_0$ included. So the prevalence is recovered only if dilution is modeled, and even then it depends on the curve at the smallest fractions, where calibration pools of confirmed positives would remove the remaining bias.

\begin{table}
\caption{Mean of $\widehat\beta_0$ minus the individual-data estimate in the same assignment, the paired difference between BADGER (P1) and LOG (Monte Carlo standard error), and the weighted curve error ($\times100$). Monte Carlo standard errors of the means are $0.01$ to $0.02$. Accurate: $\Se=\Sp=0.99$; corr.: correlated retest misses.}
\label{tab:app}
\centering\small
\input{tables/mixed_hev}
\end{table}

\begin{table}
\caption{Coverage of the individual-data estimate $\widehat\beta_0^{\mathrm{IND}}$ by $95\%$ credible intervals of BADGER, and frequencies of strong evidence from the Bayes factor ($B_{10}>10$ for dilution, $B_{10}<1/10$ for none).}
\label{tab:appbf}
\centering\small
\input{tables/mixed_hev_bf}
\end{table}

%% file: tables/mixed_hev.tex
\begin{tabular}{llccccccc}
\toprule
 & & \multicolumn{4}{c}{BADGER} & & & BADGER \\
\cmidrule(lr){3-6}
Curve & Assay & P1 & P2 & grid & sep & LOG & GT$_0$ & $-$ LOG \\
\midrule
\multicolumn{9}{l}{\emph{Mean of $\widehat\beta_0$ minus the individual-data estimate}} \\
Logistic & Accurate & +0.02 & +0.01 & +0.03 & +0.02 & +0.01 & $-$0.21 & +0.011 (0.002) \\
Logistic & VIRCLIA & +0.05 & +0.02 & +0.04 & +0.04 & +0.04 & $-$0.21 & +0.002 (0.002) \\
Threshold & Accurate & +0.13 & +0.14 & +0.16 & +0.14 & +0.16 & $-$0.56 & $-$0.032 (0.004) \\
Threshold & VIRCLIA & +0.12 & +0.12 & +0.13 & +0.13 & +0.17 & $-$0.67 & $-$0.051 (0.004) \\
Threshold, corr. & Accurate & +0.14 & +0.16 & +0.18 & +0.17 & +0.17 & $-$0.55 & $-$0.029 (0.004) \\
\midrule
\multicolumn{9}{l}{\emph{Weighted curve error ($\times100$)}} \\
Logistic & Accurate & 4.4 & 4.6 & 4.8 & 5.2 & 4.0 & 17.4 & --- \\
Logistic & VIRCLIA & 4.0 & 4.7 & 4.3 & 4.7 & 4.0 & 15.7 & --- \\
Threshold & Accurate & 8.2 & 9.3 & 10.0 & 9.2 & 9.8 & 41.8 & --- \\
Threshold & VIRCLIA & 7.6 & 9.1 & 9.9 & 8.9 & 8.8 & 37.8 & --- \\
Threshold, corr. & Accurate & 7.8 & 9.0 & 9.3 & 8.6 & 9.7 & 41.8 & --- \\
\bottomrule
\end{tabular}

%% file: tables/mixed_hev_bf.tex
\begin{tabular}{llcccccccc}
\toprule
 & & \multicolumn{4}{c}{Coverage of $\widehat\beta_0^{\mathrm{IND}}$} & \multicolumn{2}{c}{$B_{10}>10$} & \multicolumn{2}{c}{$B_{10}<1/10$} \\
\cmidrule(lr){3-6}\cmidrule(lr){7-8}\cmidrule(lr){9-10}
Curve & Assay & P1 & P2 & grid & sep & P1 & P2 & P1 & P2 \\
\midrule
Logistic & Accurate & 1.00 & 1.00 & 1.00 & 1.00 & 0.44 & 0.57 & 0.03 & 0.00 \\
Logistic & VIRCLIA & 0.98 & 0.99 & 0.98 & 0.98 & 0.15 & 0.28 & 0.11 & 0.00 \\
Threshold & Accurate & 0.90 & 0.87 & 0.87 & 0.88 & 1.00 & 1.00 & 0.00 & 0.00 \\
Threshold & VIRCLIA & 0.95 & 0.95 & 0.92 & 0.95 & 1.00 & 1.00 & 0.00 & 0.00 \\
Threshold, corr. & Accurate & 0.90 & 0.90 & 0.87 & 0.90 & 1.00 & 1.00 & 0.00 & 0.00 \\
\bottomrule
\end{tabular}

%% file: sections/06-conclusion.tex
\section{Conclusion}
\label{sec:discussion}

We have proposed BADGER, a Bayesian regression model for group testing data in which pooled sensitivity can be any nondecreasing function of the fraction of positives in a pool, and the same curve is shared by pools of different sizes. The curve is a Dirichlet mixture of distribution functions, on either a grid or a Bernstein basis. By adding the individual statuses and a latent cause for each pooled reading to the data, all full conditionals become Dirichlet or beta except the one for $\bbeta$, so the model can be fitted with a simple Gibbs sampler, and a Bayes factor computed by bridge sampling tests whether dilution is present. This flexibility has a price. Pooled data cannot tell a higher prevalence from a lower pooled sensitivity, so the intercept is learned only through the ends of the curve, its monotonicity, fractions shared by pools of different sizes, and differences in covariates between pools (\cref{prop:shared}). The slopes are not affected. In our simulations BADGER estimated the curve as well as or better than the parametric models in most settings with dilution, largely thanks to its prior, and it removed most of the bias in the intercept that comes from ignoring dilution. If dilution is known to be mild and of a known shape, a one-parameter model is enough. Otherwise, or when the curve itself is of interest, we recommend BADGER, and when the prevalence is the main target, we suggest reporting the fit under an informed prior next to the default.

The work has several limitations that suggest directions for future research. We assume that the curve depends on the pool only through the fraction of positives, which fails when detection depends on the titer of the pool. Separate curves for each pool size avoid this assumption at some cost in precision, and a viral-load model of pooling \citep{Nguyen2019} could relax it or provide an informed prior. The data carry little information about the assay sensitivity, so its prior should come from a validation study. They also carry little information about the curve at the smallest fractions, and calibration pools with known numbers of positives would help there and sharpen the estimate of the intercept. We also assume that retests are independent given the true statuses. Finally, we consider only Dorfman testing. Hierarchical protocols fit the same augmentation with a larger latent state, but array testing breaks the pool-by-pool factorization and needs a different sampler. We know of no public data set with pooled outcomes, retests and covariates, so records from a pooled screening program would be the natural next application.

%% file: sections/appendix.tex
%
\section{Derivations}
\label{app:A}

We use the notation of \mref{sec:model} and \mref{sec:bayes}\mainof{}; in particular $P_j(\bm u)\psi_j(\bm u)=\prod_i\alpha_{ij}^{1-u_i}\gamma_{ij}^{u_i}$, with $\alpha_{ij}$ and $\gamma_{ij}$ as in \mref{eq:pbrec}.

\subsection{The Status Draw}
\label{sec:exact}

We show that the backward draw of \mref{sec:gibbs}\mainof{} is an exact draw from \mref{eq:condY}. Fix a pool and drop the index $j$. Unrolling \mref{eq:pbrec} gives
\[
q^{(i)}(k)=\sum_{u_1+\dots+u_i=k}\ \prod_{l=1}^{i}\alpha_l^{1-u_l}\gamma_l^{u_l},
\]
the total weight of the configurations of the first $i$ members with $k$ positives; $q^{(c)}=q$, and $q^{(i)}(k)=0$ for $k<0$ or $k>i$.

The target is $\Prob(\bm u)\propto\prod_l\alpha_l^{1-u_l}\gamma_l^{u_l}$ on $\{|\bm u|=k\}$. Given $u_{i+1},\dots,u_c$, let $k_i=k-\sum_{l>i}u_l$ be the positives left for members $1,\dots,i$. Summing the target over $u_1,\dots,u_{i-1}$ with $u_1+\dots+u_{i-1}=k_i-u_i$ gives $\Prob(u_i\mid u_{i+1},\dots,u_c)\propto\alpha_i^{1-u_i}\gamma_i^{u_i}q^{(i-1)}(k_i-u_i)$. By \mref{eq:pbrec} the two values sum to $q^{(i)}(k_i)$, so
\[
\Prob(u_i=1\mid u_{i+1},\dots,u_c)=\frac{\gamma_i\,q^{(i-1)}(k_i-1)}{q^{(i)}(k_i)},
\]
which is the rule of \mref{sec:gibbs}\mainof{}. Moves to states that cannot reach $(0,0)$ have probability $0$ (\mref{fig:backward}\mainof{}). Multiplying the conditionals, with $k_{i-1}=k_i-u_i$ and $k_c=k$, the ratios telescope:
\[
\Prob(\bm u)=\prod_{i=1}^{c}\frac{\alpha_i^{1-u_i}\gamma_i^{u_i}\,q^{(i-1)}(k_{i-1})}{q^{(i)}(k_i)}
=\frac{q^{(0)}(k_0)}{q(k)}\prod_{i=1}^{c}\alpha_i^{1-u_i}\gamma_i^{u_i}
=\frac{P(\bm u)\psi(\bm u)\,\mathbf 1\{|\bm u|=k\}}{q(k)},
\]
since $q^{(0)}(k_0)=\mathbf 1\{k_0=0\}$. This is \mref{eq:condY}. The table $q^{(i)}$ is already computed for \mref{eq:condK}, so the draw adds $O(c)$ operations.

\subsection{Proof of the Shared-Fractions Result}
\label{sec:identproof}

With no covariates, a pool of size $c$ has $P_j(\bu)=p^{|\bu|}(1-p)^{c-|\bu|}$. Write $\theta_{k,c}=h(k,c)\,p^k(1-p)^{c-k}$.

\emph{The data determine every $\theta_{k,c}$.} For a pool of size $c$ with retests $\by$,
\[
\Prob(Z_j=1,\bY_j=\by)=\sum_{k=0}^{c}\theta_{k,c}\,S_k(\by),\qquad S_k(\by)=\sum_{|\bu|=k}\prod_if(y_i\mid u_i),
\]
and $\Prob(Z_j=0)=1-\sum_k\binom ck\theta_{k,c}$, so the data depend on $(p,h)$ only through the $\theta_{k,c}$. Given statuses $\bu$ with $|\bu|=k$, the retests are independent Bernoulli$(\Se)$ for positive and Bernoulli$(1-\Sp)$ for negative members, so
\[
\E\bigl(t^{|\bY_j|}\mid\bu\bigr)=a(t)^kb(t)^{c-k},\qquad a(t)=1-\Se+\Se\,t,\quad b(t)=\Sp+(1-\Sp)\,t,
\]
which depends on $\bu$ only through $k$. Each of the $\binom ck$ terms of $S_k$ therefore contributes $a^kb^{c-k}$ to the generating function, and
\[
\sum_{s=0}^{c}\Prob(Z_j=1,|\bY_j|=s)\,t^s=\sum_{k=0}^{c}\binom ck\theta_{k,c}\,a(t)^kb(t)^{c-k}.
\]
Since $\Se\,\Sp-(1-\Se)(1-\Sp)=\Se+\Sp-1\ne0$, $a$ and $b$ are linearly independent, so $\{a^kb^{c-k}\}_{k=0}^{c}$ is a basis of the polynomials of degree at most $c$, and the left side determines $\theta_{0,c},\dots,\theta_{c,c}$.

\emph{The shared fraction determines $p$.} By \mref{eq:hF}, the cells $(k,c)$ and $(k',c')$ with $k/c=k'/c'=x$ share the pooled sensitivity, so
\[
\frac{\theta_{k',c'}}{\theta_{k,c}}=\frac{p^{k'}(1-p)^{c'-k'}}{p^k(1-p)^{c-k}}=\{p^x(1-p)^{1-x}\}^{c'-c}.
\]
If $(p',h')$ gives the same $\theta$'s, then $h'(k,c)>0$ and the two ratios agree, so $p^x(1-p)^{1-x}=p'^x(1-p')^{1-x}$. Since $x\log p+(1-x)\log(1-p)$ is increasing on $(0,x)$, $p=p'$, and then $h'(k,c)=h(k,c)$ for every pool size in the data and $k=0,\dots,c$. \hfill$\square$

\emph{One pool size is not enough.} With a single size $c$, every $p'\in(0,1)$ reproduces $\theta_{1,c},\dots,\theta_{c-1,c}$ with $h'(k,c)=\theta_{k,c}/\{p'^k(1-p')^{c-k}\}$, so $p$ is identified only through the ends, $\theta_{0,c}=\varphi(1-p)^c$ and $\theta_{c,c}=\Se\,p^c$, and through the constraint that $h'$ be nondecreasing in $[\varphi,\Se]$.

\subsection{Expected Cell Counts and the Bernstein Basis}
\label{sec:special}

\paragraph{Expected cell counts} In the design of \mref{sec:sim}\mainof{} the marginal prevalence is $0.078$, so $K_j\sim\mathrm{Bin}(c_j,0.078)$. The expected numbers of pools of size $5$ with $k=1,\dots,4$ positives are $84.5$, $14.3$, $1.21$ and $0.05$, and of size $10$ with $k=1,\dots,5$ they are $56.3$, $21.4$, $4.83$, $0.71$ and $0.07$; these are the weights of the curve error. With $\Se=\Sp=0.99$ the parametric submodels \mref{eq:param} stay in $[\varphi,\Se]$, and so lie in the grid model, for $\lambda$ below $10.2$ (logistic), $5.2$ (probit) and $6.8$ (complementary log-log), which covers every value used in \mref{sec:sim}\mainof{}.

\paragraph{Bernstein basis} A continuous nondecreasing $F$ with $F(1)=1$ is the uniform limit of its Bernstein polynomials $\sum_{r=0}^{L-1}F\{r/(L-1)\}\binom{L-1}{r}x^r(1-x)^{L-1-r}$, which equal $\sum_mw_mG_m$ with $w_1=F(0^+)$ and $w_{r+1}=F\{r/(L-1)\}-F\{(r-1)/(L-1)\}\ge0$. The basis is therefore dense in these curves, and $w_1$ allows a jump at $0^+$ for strongly positive specimens. The prior mean $\{1+(L-1)x\}/L$ follows from $\sum_{m=2}^{L}G_m(x)=(L-1)x$, the mean of a Binomial$(L-1,x)$ variable written as a sum of tail probabilities. Each $G_m$ with $m\ge2$ covers the whole interval, so pools at every fraction inform every weight.

\section{Additional results}
\label{app:B}

\input{sections/appendix_results}

%% file: sections/appendix_results.tex
\subsection{Checks of the Bayes Factor Estimates}
\label{sec:bfcheck}

We compared the bridge sampling estimate of $m_1$ (\mref{sec:bf}\mainof{}) with importance sampling from a different proposal, using $50{,}000$ joint draws of $(\bbeta,\bm w)$. The $240$ data sets come from the mixed design of \mref{sec:sim}\mainof{} (no dilution, moderate linear and severe convex shapes, both bases) and from single-size designs ($N=2000$, $c\in\{5,10\}$, prevalence $0.02$ or $0.20$). The two estimates agreed to within $0.08$ on the $\log_{10}$ scale in every data set, with mean absolute differences of $0.006$ to $0.020$, about their standard errors (Supplementary Table~\ref{tab:bfcheck}).

\begin{table}[!ht]
\caption{Bayes factor estimates against importance sampling (IS): median $\log_{10}\widehat B_{10}$ by bridge sampling, mean and maximum absolute difference from IS, and the median standard error of the bridge estimate; $20$ data sets per design.}
\label{tab:bfcheck}
\centering\small
\input{tables/bf_check}
\end{table}

\subsection{Mixed Pool Sizes: Full Results}
\label{sec:mixedfull}

Supplementary Tables~\ref{tab:mixederr}--\ref{tab:mixeddiag} give the full results of the mixed study of \mref{sec:sim}\mainof{}. The weighted curve error is \mref{eq:curveerr}\mainof{}, which weights cell $(k,c)$ by its expected number of pools (Supplement~\ref{sec:special}); the unweighted error sets $\omega_{kc}=1$ for the $13$ cells $1\le k<c$. Four points add to the main text.

\paragraph{Parametric fits} CLL had the smallest curve error under the linear and convex--concave shapes at the severe level, because over the populated fractions it rises more slowly than LOG.

\paragraph{Bases and separate curves} The grid beat the Bernstein basis under the moderate and severe convex and the severe convex--concave shapes, all convex at the smallest fractions, where the grid's steps sit; $L=4$ or $11$ did not close the gap. The Bernstein basis was clearly better under the severe logistic and the moderate and severe Hung--Swallow shapes. Separate curves beat the shared grid curve only under the last two, which are steepest at the smallest fractions.

\paragraph{The Bayes factor} Under the Hung--Swallow shape $B_{10}>10$ occurred in $34\%$, $86\%$ and $96\%$ of replications at the three levels (Supplementary Table~\ref{tab:mixedbf}). Under the other shapes, $B_{10}<1/10$ occurred in $55\%$ to $72\%$ of mild and $4\%$ to $19\%$ of moderate replications, so neither prior separates no dilution from mild dilution in this design. The Bayes factor and the likelihood ratio statistic ranked the null data sets almost identically (rank correlation $0.99$).

\paragraph{Convergence} Mixing was slowest for $h(1,10)$ under the Bernstein basis (Supplementary Table~\ref{tab:mixeddiag}): with moderate or severe dilution, up to $15\%$ of fits had $\widehat R>1.05$, with the two chains at different points of the ridge. Dropping these fits changed the mean of $\widehat\beta_0$ by at most $0.006$. The Metropolis--Hastings step for $\bbeta$ accepted over $92\%$ of proposals in every fit.

\begin{table}[!ht]
\caption{Weighted curve error ($\times100$, Monte Carlo standard error); mixed design, $200$ replications. Smallest error in each row in bold. The last column is GT$_0$ ($h\equiv\Se$).}
\label{tab:mixederr}
\centering\scriptsize
\input{tables/mixed_err}
\end{table}

\begin{table}[!ht]
\caption{Unweighted curve error over the $13$ cells $1\le k<c$, as Supplementary Table~\ref{tab:mixederr}.}
\label{tab:mixederr13}
\centering\scriptsize
\input{tables/mixed_err13}
\end{table}

\begin{table}[!ht]
\caption{Bias of $\widehat\beta_0$ (standard deviation over replications; divide by $\sqrt{200}$ for the Monte Carlo standard error); $\beta_0=-3$, mixed design. Smallest absolute bias in each row in bold. Last row: root mean squared error averaged over the $22$ settings.}
\label{tab:mixedbeta}
\centering\scriptsize
\input{tables/mixed_beta0}
\end{table}

\begin{table}[!ht]
\caption{Coverage of $95\%$ credible intervals for $\bbeta$ and for $h$ at the populated cells $k/c\le0.4$; mixed design. Monte Carlo standard error about $0.015$.}
\label{tab:mixedcov}
\centering\scriptsize
\input{tables/mixed_cov}
\end{table}

\begin{table}[!ht]
\caption{The Bayes factor for dilution, mixed design: frequencies of $B_{10}<1/10$ and $B_{10}>10$ and the median $\log_{10}\widehat B_{10}$ under P1 (both bases) and P2; rejection rate of the likelihood ratio test (LRT) of $\lambda=0$ in LOG at the nominal $5\%$ level ($\tfrac12\chi^2_0+\tfrac12\chi^2_1$ reference); and detection rates of $B_{10}$ and the LRT statistic at their own $95\%$ null quantiles. Monte Carlo standard errors at most $0.035$.}
\label{tab:mixedbf}
\centering\scriptsize
\input{tables/mixed_bf}
\end{table}

\begin{table}[!ht]
\caption{Convergence: median effective sample size of $\beta_0$, $h(1,5)$ and $h(1,10)$ ($3000$ draws) and the largest split-$\widehat R$ over $\beta_0$, $\beta_1$, $h(1,5)$, $h(1,10)$ and all replications; for default BADGER also the $5\%$ quantile of the effective sample size of $h(1,10)$ and the share of fits with $\widehat R>1.05$.}
\label{tab:mixeddiag}
\centering\scriptsize
\input{tables/mixed_diag}
\end{table}

\subsection{Pooled Specificity and Unknown Accuracy}
\label{sec:accfull}

Both studies use the mixed design, the Bernstein basis and the concave, linear and logistic shapes at the moderate level. Retests inform $\Se$ and $\Sp$, but $\varphi$ is informed only by the few positive pools whose retests are all negative, a pattern that a missed positive also produces. Once $\varphi$ is free, the data separate it from $\Se$ only weakly.

\paragraph{Pooled specificity} With $\varphi=0.03$ against $1-\Sp=0.01$ and a uniform prior on $\rho$, BADGER recovered $\varphi$ (posterior means $0.033$ to $0.035$, coverage $95\%$; Supplementary Table~\ref{tab:pooledsp}). Learning $\varphi$ raised $\widehat\beta_0$ by $0.02$ to $0.03$ and changed the curve error by at most $0.7$ ($\times100$). LOG, which assumes $\varphi=1-\Sp$, was biased by $-0.03$ to $-0.08$.

\paragraph{Unknown accuracy} With $\Se=0.95$ and $\Sp=0.98$, BADGER was fitted with the accuracy fixed, correctly or at $0.99/0.99$, or learned with $\varphi$ under beta priors of prior sample size $n_0\in\{100,20\}$ centered at the truth or at $0.99$ (Supplementary Table~\ref{tab:accuracy}). The posterior mean of $\Se$ stayed within $0.003$ of a prior mean of $0.99$, although its chains mixed well (largest $\widehat R$ $1.04$), so the likelihood carries little information about it. Learning the accuracy widened the intervals for $\beta_0$ by $21\%$ to $37\%$ but did not reduce the bias, and with $n_0=20$ it raised $\widehat\beta_0$ by up to $0.07$. Under the logistic shape all fits were biased downward, because in the parametrization of \citet{Warasi2017} this shape is much more severe at $\Se=0.95$. With the same individuals in $600$ pools of five, the posterior of $(\Se,\Sp)$ was as good but the bias of $\widehat\beta_0$ grew by $0.03$ to $0.05$: mixed sizes help the intercept through the shared curve, not by separating prevalence from testing error as in \citet{Huang2017}.

\begin{table}[!ht]
\caption{Pooled false-positive rate $0.03$ against $\Sp=0.99$: bias (Monte Carlo standard error) and standard deviation of $\widehat\beta_0$, coverage, weighted curve error, and the mean posterior mean of $\varphi$ with its coverage; moderate level, $200$ replications.}
\label{tab:pooledsp}
\centering\footnotesize
\input{tables/mixed_pooledsp}
\end{table}

\begin{table}[!ht]
\caption{Assay with $\Se=0.95$, $\Sp=0.98$: BADGER with the accuracy fixed, correctly or wrongly, or learned under beta priors of prior sample size $n_0$ centered at the truth or at $0.99/0.99$. The last row of each block uses $600$ pools of size $5$. Bias (Monte Carlo standard error), coverage and mean width of intervals for $\beta_0$, weighted curve error, and bias and coverage for $\Se$ and $\Sp$; moderate level, $200$ replications.}
\label{tab:accuracy}
\centering\footnotesize
\input{tables/mixed_accuracy}
\end{table}

\subsection{Prior Sensitivity}
\label{sec:priorsens}

Supplementary Table~\ref{tab:priorimplied} shows that P1 is not neutral: as a sparse prior on curves with few steps, it puts $\Prob\{h(1,10)<0.1\}=0.78$ under the grid and $0.42$ under the Bernstein basis.

\begin{table}[!ht]
\caption{Curve at the smallest fractions implied by the priors ($\Se=\Sp=0.99$): mean, quartiles and tail probabilities of $h(1,10)$ and $h(1,5)$ under P1 (grid and Bernstein, $L=7$) and P2; one million prior draws.}
\label{tab:priorimplied}
\centering\small
\input{tables/prior_implied}
\end{table}

Supplementary Table~\ref{tab:prior} compares P1 with P2. Under no dilution and at the mild level, P2 lowered $\widehat\beta_0$ by $0.02$ to $0.03$, which removed most of P1's bias. At the moderate and severe levels it had the larger curve error in $13$ of $14$ settings, and its root mean squared error over all settings was $0.155$, against $0.150$. As \mref{sec:ident}\mainof{} explains, a prior that favors less dilution favors a lower intercept.

\begin{table}[!ht]
\caption{Prior sensitivity: bias, root mean squared error and coverage of $\widehat\beta_0$, and weighted curve error ($\times100$), under P1 and P2; mixed design, $200$ replications.}
\label{tab:prior}
\centering\scriptsize
\input{tables/mixed_prior}
\end{table}

%% file: tables/bf_check.tex
\begin{tabular}{lccccc}
\toprule
Design & $R$ & median $\log_{10}\widehat B_{10}$ & \multicolumn{2}{c}{$|$bridge $-$ IS$|$: mean, max} & median SE, bridge \\
\midrule
mixed: None None bernstein & 20 & $-$1.41 & 0.007 & 0.027 & 0.007 \\
mixed: None None grid & 20 & $-$1.68 & 0.009 & 0.019 & 0.007 \\
mixed: Linear Moderate bernstein & 20 & $-$0.08 & 0.011 & 0.043 & 0.008 \\
mixed: Linear Moderate grid & 20 & $-$0.32 & 0.017 & 0.045 & 0.009 \\
mixed: Convex Severe bernstein & 20 & +2.06 & 0.010 & 0.021 & 0.011 \\
mixed: Convex Severe grid & 20 & +1.49 & 0.020 & 0.076 & 0.013 \\
single: c=5 prev=0.02 none & 20 & $-$0.60 & 0.006 & 0.016 & 0.006 \\
single: c=10 prev=0.02 none & 20 & $-$1.10 & 0.007 & 0.017 & 0.007 \\
single: c=5 prev=0.20 none & 20 & $-$1.46 & 0.006 & 0.021 & 0.005 \\
single: c=10 prev=0.20 none & 20 & $-$2.14 & 0.013 & 0.042 & 0.009 \\
single: c=5 prev=0.02 moderate & 20 & $-$0.43 & 0.008 & 0.038 & 0.006 \\
single: c=10 prev=0.20 moderate & 20 & +6.41 & 0.020 & 0.065 & 0.012 \\
\bottomrule
\end{tabular}

%% file: tables/mixed_err.tex
\begin{tabular}{llccccccc}
\toprule
Shape & Level & BADGER & BADGER-grid & BADGER-sep & LOG & PB & CLL & GT$_0$ \\
\midrule
None & --- & 7.0\,(0.2) & 7.2\,(0.2) & 9.5\,(0.2) & \textbf{3.3\,(0.3)} & 3.4\,(0.3) & 3.4\,(0.3) & 0.0 \\
Concave & Mild & 4.9\,(0.3) & 5.1\,(0.3) & 7.1\,(0.3) & \textbf{4.7\,(0.3)} & 4.9\,(0.3) & 5.0\,(0.3) & 4.6 \\
Concave & Moderate & 5.5\,(0.3) & \textbf{5.5\,(0.3)} & 6.1\,(0.2) & 6.2\,(0.3) & 6.3\,(0.3) & 6.5\,(0.3) & 16.0 \\
Concave & Severe & \textbf{6.0\,(0.3)} & 6.5\,(0.3) & 7.0\,(0.2) & 6.2\,(0.3) & 6.0\,(0.3) & 6.1\,(0.3) & 27.5 \\
Linear & Mild & \textbf{4.4\,(0.3)} & 4.5\,(0.3) & 6.3\,(0.3) & 5.3\,(0.2) & 5.4\,(0.2) & 5.6\,(0.3) & 6.5 \\
Linear & Moderate & \textbf{5.9\,(0.3)} & 6.1\,(0.3) & 6.5\,(0.3) & 6.9\,(0.3) & 6.6\,(0.3) & 6.6\,(0.4) & 22.6 \\
Linear & Severe & 6.9\,(0.3) & 6.9\,(0.3) & 7.6\,(0.3) & 8.1\,(0.4) & 7.1\,(0.4) & \textbf{6.5\,(0.4)} & 38.8 \\
Convex & Mild & \textbf{4.2\,(0.3)} & 4.2\,(0.3) & 5.9\,(0.3) & 5.5\,(0.2) & 5.6\,(0.2) & 5.7\,(0.2) & 7.6 \\
Convex & Moderate & 6.3\,(0.3) & \textbf{6.0\,(0.3)} & 6.9\,(0.2) & 8.2\,(0.3) & 7.5\,(0.3) & 7.1\,(0.3) & 26.8 \\
Convex & Severe & 8.8\,(0.3) & \textbf{7.7\,(0.4)} & 8.6\,(0.3) & 11.6\,(0.4) & 10.1\,(0.4) & 8.6\,(0.4) & 45.9 \\
Concave--convex & Mild & \textbf{4.6\,(0.3)} & 4.7\,(0.3) & 6.6\,(0.3) & 5.1\,(0.2) & 5.3\,(0.3) & 5.4\,(0.3) & 5.6 \\
Concave--convex & Moderate & \textbf{6.1\,(0.3)} & 6.1\,(0.3) & 6.3\,(0.2) & 7.1\,(0.3) & 7.0\,(0.3) & 7.1\,(0.4) & 19.5 \\
Concave--convex & Severe & \textbf{6.3\,(0.3)} & 6.4\,(0.3) & 7.1\,(0.3) & 7.7\,(0.4) & 6.9\,(0.4) & 6.4\,(0.4) & 33.5 \\
Convex--concave & Mild & \textbf{4.3\,(0.3)} & 4.4\,(0.3) & 6.1\,(0.3) & 5.3\,(0.2) & 5.5\,(0.2) & 5.6\,(0.3) & 7.3 \\
Convex--concave & Moderate & \textbf{6.0\,(0.3)} & 6.2\,(0.3) & 6.9\,(0.2) & 7.2\,(0.3) & 6.7\,(0.3) & 6.5\,(0.3) & 25.5 \\
Convex--concave & Severe & 7.0\,(0.3) & \textbf{6.6\,(0.3)} & 7.7\,(0.3) & 9.1\,(0.3) & 7.8\,(0.3) & 6.6\,(0.3) & 43.8 \\
Logistic & Mild & \textbf{4.7\,(0.3)} & 4.8\,(0.3) & 6.6\,(0.3) & 5.3\,(0.2) & 5.5\,(0.3) & 5.7\,(0.3) & 6.8 \\
Logistic & Moderate & 5.8\,(0.3) & 6.0\,(0.3) & 6.5\,(0.2) & \textbf{5.7\,(0.3)} & 6.0\,(0.3) & 6.3\,(0.3) & 17.4 \\
Logistic & Severe & 5.5\,(0.3) & 7.5\,(0.3) & 7.8\,(0.2) & \textbf{5.2\,(0.3)} & 5.4\,(0.3) & 6.2\,(0.3) & 36.2 \\
Hung--Swallow & Mild & 7.3\,(0.3) & 7.4\,(0.3) & 7.4\,(0.3) & \textbf{6.4\,(0.3)} & 7.0\,(0.3) & 7.6\,(0.3) & 19.8 \\
Hung--Swallow & Moderate & 6.1\,(0.2) & 8.7\,(0.2) & 8.3\,(0.2) & \textbf{6.0\,(0.2)} & 6.9\,(0.2) & 8.0\,(0.2) & 32.3 \\
Hung--Swallow & Severe & \textbf{5.3\,(0.3)} & 8.4\,(0.2) & 8.1\,(0.2) & 5.9\,(0.2) & 6.8\,(0.2) & 8.2\,(0.2) & 41.1 \\
\bottomrule
\end{tabular}

%% file: tables/mixed_err13.tex
\begin{tabular}{llccccccc}
\toprule
Shape & Level & BADGER & BADGER-grid & BADGER-sep & LOG & PB & CLL & GT$_0$ \\
\midrule
None & --- & 4.1\,(0.1) & 4.4\,(0.1) & 5.6\,(0.1) & \textbf{1.6\,(0.1)} & 1.7\,(0.2) & 1.8\,(0.2) & 0.0 \\
Concave & Mild & 3.0\,(0.2) & 3.4\,(0.2) & 4.4\,(0.2) & \textbf{2.2\,(0.1)} & 2.3\,(0.1) & 2.5\,(0.1) & 2.6 \\
Concave & Moderate & \textbf{3.2\,(0.2)} & 3.8\,(0.2) & 4.2\,(0.2) & 3.8\,(0.1) & 3.6\,(0.1) & 3.5\,(0.2) & 8.9 \\
Concave & Severe & \textbf{3.7\,(0.2)} & 4.8\,(0.2) & 4.9\,(0.2) & 5.4\,(0.1) & 4.5\,(0.1) & 3.9\,(0.2) & 15.3 \\
Linear & Mild & \textbf{2.4\,(0.2)} & 2.8\,(0.2) & 3.7\,(0.2) & 2.7\,(0.1) & 2.8\,(0.1) & 2.9\,(0.1) & 4.0 \\
Linear & Moderate & \textbf{3.6\,(0.2)} & 4.0\,(0.2) & 4.1\,(0.2) & 6.4\,(0.1) & 5.6\,(0.2) & 4.9\,(0.2) & 14.0 \\
Linear & Severe & 5.9\,(0.2) & \textbf{5.0\,(0.2)} & 5.2\,(0.2) & 10.4\,(0.2) & 8.7\,(0.2) & 6.7\,(0.2) & 24.0 \\
Convex & Mild & \textbf{2.4\,(0.1)} & 2.5\,(0.1) & 3.2\,(0.1) & 3.7\,(0.1) & 3.6\,(0.1) & 3.6\,(0.1) & 5.5 \\
Convex & Moderate & 6.1\,(0.2) & \textbf{5.0\,(0.2)} & 5.7\,(0.2) & 10.7\,(0.2) & 9.6\,(0.2) & 8.6\,(0.2) & 19.3 \\
Convex & Severe & 10.5\,(0.3) & \textbf{7.7\,(0.3)} & 8.8\,(0.3) & 18.0\,(0.2) & 15.9\,(0.2) & 13.2\,(0.2) & 33.0 \\
Concave--convex & Mild & \textbf{2.5\,(0.2)} & 2.7\,(0.2) & 3.6\,(0.2) & 3.1\,(0.1) & 3.1\,(0.1) & 3.2\,(0.1) & 4.0 \\
Concave--convex & Moderate & 5.0\,(0.2) & 4.8\,(0.2) & \textbf{4.7\,(0.1)} & 7.8\,(0.2) & 7.3\,(0.2) & 6.8\,(0.2) & 14.1 \\
Concave--convex & Severe & 7.9\,(0.2) & \textbf{6.3\,(0.2)} & 6.8\,(0.2) & 12.7\,(0.2) & 11.3\,(0.2) & 9.7\,(0.2) & 24.2 \\
Convex--concave & Mild & 2.7\,(0.2) & 3.1\,(0.2) & 4.0\,(0.2) & \textbf{2.5\,(0.1)} & 2.6\,(0.1) & 2.8\,(0.1) & 4.0 \\
Convex--concave & Moderate & 4.1\,(0.1) & 4.9\,(0.2) & 5.0\,(0.1) & 5.4\,(0.1) & 4.4\,(0.2) & \textbf{3.7\,(0.2)} & 14.0 \\
Convex--concave & Severe & 5.8\,(0.2) & 6.8\,(0.2) & 6.5\,(0.1) & 8.7\,(0.2) & 6.7\,(0.2) & \textbf{4.5\,(0.2)} & 24.1 \\
Logistic & Mild & 3.3\,(0.2) & 3.7\,(0.2) & 4.7\,(0.2) & \textbf{2.4\,(0.1)} & 2.6\,(0.1) & 2.8\,(0.1) & 3.3 \\
Logistic & Moderate & 4.5\,(0.2) & 5.6\,(0.2) & 6.0\,(0.2) & \textbf{2.4\,(0.1)} & 2.8\,(0.1) & 3.5\,(0.1) & 7.7 \\
Logistic & Severe & 5.7\,(0.2) & 9.5\,(0.3) & 9.6\,(0.2) & \textbf{2.2\,(0.1)} & 3.1\,(0.1) & 5.1\,(0.2) & 15.4 \\
Hung--Swallow & Mild & 5.5\,(0.2) & 7.2\,(0.3) & 7.9\,(0.2) & \textbf{2.9\,(0.1)} & 3.8\,(0.1) & 4.8\,(0.2) & 8.0 \\
Hung--Swallow & Moderate & 5.0\,(0.2) & 9.5\,(0.3) & 9.9\,(0.2) & \textbf{2.8\,(0.1)} & 3.8\,(0.1) & 5.7\,(0.2) & 13.8 \\
Hung--Swallow & Severe & 4.4\,(0.2) & 9.4\,(0.3) & 10.2\,(0.3) & \textbf{3.2\,(0.1)} & 3.7\,(0.1) & 6.0\,(0.2) & 18.4 \\
\bottomrule
\end{tabular}

%% file: tables/mixed_beta0.tex
\begin{tabular}{llccccccc}
\toprule
Shape & Level & BADGER & BADGER-grid & BADGER-sep & LOG & PB & CLL & GT$_0$ \\
\midrule
None & --- & +0.07\,(0.10) & +0.08\,(0.10) & +0.10\,(0.10) & +0.04\,(0.11) & +0.04\,(0.11) & +0.04\,(0.11) & \textbf{\phantom{+}0.00\,(0.10)} \\
Concave & Mild & +0.05\,(0.12) & +0.05\,(0.12) & +0.07\,(0.11) & \textbf{+0.02\,(0.12)} & +0.02\,(0.13) & +0.02\,(0.13) & $-$0.05\,(0.11) \\
Concave & Moderate & +0.02\,(0.13) & +0.03\,(0.13) & +0.04\,(0.13) & \textbf{\phantom{+}0.00\,(0.14)} & \phantom{+}0.00\,(0.14) & +0.01\,(0.14) & $-$0.18\,(0.12) \\
Concave & Severe & \textbf{+0.01\,(0.16)} & +0.02\,(0.17) & +0.02\,(0.16) & $-$0.02\,(0.15) & $-$0.01\,(0.16) & +0.01\,(0.16) & $-$0.34\,(0.12) \\
Linear & Mild & +0.04\,(0.12) & +0.04\,(0.12) & +0.06\,(0.12) & \textbf{+0.01\,(0.13)} & +0.01\,(0.13) & +0.01\,(0.13) & $-$0.07\,(0.12) \\
Linear & Moderate & \textbf{\phantom{+}0.00\,(0.15)} & +0.01\,(0.15) & +0.02\,(0.15) & $-$0.03\,(0.15) & $-$0.02\,(0.15) & \phantom{+}0.00\,(0.16) & $-$0.27\,(0.12) \\
Linear & Severe & $-$0.06\,(0.17) & \textbf{$-$0.03\,(0.18)} & $-$0.04\,(0.17) & $-$0.10\,(0.16) & $-$0.07\,(0.17) & $-$0.04\,(0.17) & $-$0.54\,(0.14) \\
Convex & Mild & +0.03\,(0.12) & +0.03\,(0.12) & +0.05\,(0.12) & $-$0.01\,(0.13) & \phantom{+}0.00\,(0.13) & \textbf{\phantom{+}0.00\,(0.13)} & $-$0.09\,(0.12) \\
Convex & Moderate & $-$0.04\,(0.15) & $-$0.02\,(0.15) & \textbf{$-$0.02\,(0.15)} & $-$0.08\,(0.15) & $-$0.06\,(0.15) & $-$0.05\,(0.15) & $-$0.34\,(0.13) \\
Convex & Severe & $-$0.12\,(0.19) & \textbf{$-$0.08\,(0.20)} & $-$0.09\,(0.19) & $-$0.18\,(0.18) & $-$0.15\,(0.18) & $-$0.11\,(0.19) & $-$0.68\,(0.15) \\
Concave--convex & Mild & +0.04\,(0.12) & +0.05\,(0.12) & +0.07\,(0.12) & \textbf{+0.01\,(0.13)} & +0.01\,(0.13) & +0.02\,(0.13) & $-$0.06\,(0.11) \\
Concave--convex & Moderate & \textbf{\phantom{+}0.00\,(0.14)} & +0.01\,(0.14) & +0.02\,(0.14) & $-$0.03\,(0.15) & $-$0.02\,(0.15) & $-$0.01\,(0.16) & $-$0.23\,(0.12) \\
Concave--convex & Severe & $-$0.05\,(0.15) & \textbf{$-$0.03\,(0.16)} & $-$0.04\,(0.16) & $-$0.09\,(0.15) & $-$0.07\,(0.15) & $-$0.05\,(0.16) & $-$0.45\,(0.13) \\
Convex--concave & Mild & +0.03\,(0.13) & +0.03\,(0.13) & +0.06\,(0.12) & \textbf{\phantom{+}0.00\,(0.13)} & \phantom{+}0.00\,(0.13) & \phantom{+}0.00\,(0.13) & $-$0.08\,(0.12) \\
Convex--concave & Moderate & $-$0.01\,(0.16) & \textbf{+0.01\,(0.16)} & +0.01\,(0.15) & $-$0.04\,(0.16) & $-$0.03\,(0.16) & $-$0.01\,(0.16) & $-$0.31\,(0.13) \\
Convex--concave & Severe & $-$0.06\,(0.18) & \textbf{$-$0.01\,(0.19)} & $-$0.03\,(0.19) & $-$0.11\,(0.17) & $-$0.08\,(0.17) & $-$0.04\,(0.18) & $-$0.62\,(0.14) \\
Logistic & Mild & +0.04\,(0.12) & +0.05\,(0.12) & +0.07\,(0.12) & \textbf{+0.01\,(0.13)} & +0.02\,(0.13) & +0.02\,(0.13) & $-$0.07\,(0.12) \\
Logistic & Moderate & +0.03\,(0.14) & +0.04\,(0.14) & +0.04\,(0.13) & \textbf{+0.01\,(0.14)} & +0.02\,(0.14) & +0.03\,(0.15) & $-$0.20\,(0.12) \\
Logistic & Severe & +0.02\,(0.16) & +0.05\,(0.17) & +0.04\,(0.17) & \textbf{\phantom{+}0.00\,(0.15)} & +0.02\,(0.16) & +0.05\,(0.16) & $-$0.47\,(0.13) \\
Hung--Swallow & Mild & +0.06\,(0.14) & +0.06\,(0.14) & +0.06\,(0.14) & \textbf{+0.05\,(0.14)} & +0.06\,(0.15) & +0.07\,(0.15) & $-$0.22\,(0.11) \\
Hung--Swallow & Moderate & +0.05\,(0.14) & +0.08\,(0.17) & +0.07\,(0.16) & \textbf{+0.05\,(0.15)} & +0.07\,(0.15) & +0.09\,(0.16) & $-$0.40\,(0.13) \\
Hung--Swallow & Severe & \textbf{+0.01\,(0.16)} & +0.05\,(0.19) & +0.05\,(0.18) & +0.02\,(0.16) & +0.05\,(0.17) & +0.08\,(0.17) & $-$0.57\,(0.14) \\
\midrule
\multicolumn{2}{l}{Root mean squared error, mean over the settings} & 0.150 & 0.156 & 0.154 & 0.154 & 0.154 & 0.156 & 0.322 \\
\bottomrule
\end{tabular}

%% file: tables/mixed_cov.tex
\begin{tabular}{llcccccccc}
\toprule
 & & \multicolumn{4}{c}{BADGER (Bernstein, $L=7$)} & \multicolumn{4}{c}{BADGER-grid} \\
\cmidrule(lr){3-6}\cmidrule(lr){7-10}
Shape & Level & $\beta_0$ & $\beta_1$ & $\beta_2$ & $h$, $k/c\le0.4$ & $\beta_0$ & $\beta_1$ & $\beta_2$ & $h$, $k/c\le0.4$ \\
\midrule
None & --- & 0.92 & 0.94 & 0.98 & 0.00 & 0.93 & 0.95 & 0.98 & 0.00 \\
Concave & Mild & 0.94 & 0.94 & 0.98 & 0.97 & 0.94 & 0.93 & 0.98 & 0.99 \\
Concave & Moderate & 0.96 & 0.94 & 0.96 & 0.98 & 0.96 & 0.94 & 0.96 & 0.98 \\
Concave & Severe & 0.99 & 0.94 & 0.96 & 0.98 & 0.97 & 0.94 & 0.96 & 0.98 \\
Linear & Mild & 0.96 & 0.92 & 0.96 & 0.98 & 0.96 & 0.94 & 0.96 & 0.99 \\
Linear & Moderate & 0.95 & 0.94 & 0.98 & 0.98 & 0.94 & 0.96 & 0.98 & 0.98 \\
Linear & Severe & 0.98 & 0.95 & 0.96 & 0.98 & 0.98 & 0.96 & 0.97 & 0.98 \\
Convex & Mild & 0.96 & 0.94 & 0.94 & 0.99 & 0.96 & 0.94 & 0.96 & 0.99 \\
Convex & Moderate & 0.96 & 0.92 & 0.95 & 0.96 & 0.97 & 0.92 & 0.96 & 0.97 \\
Convex & Severe & 0.95 & 0.92 & 0.96 & 0.93 & 0.97 & 0.91 & 0.96 & 0.97 \\
Concave--convex & Mild & 0.95 & 0.92 & 0.97 & 0.97 & 0.94 & 0.93 & 0.97 & 0.99 \\
Concave--convex & Moderate & 0.96 & 0.97 & 0.96 & 0.99 & 0.96 & 0.97 & 0.95 & 0.99 \\
Concave--convex & Severe & 0.96 & 0.95 & 0.96 & 0.96 & 0.98 & 0.96 & 0.96 & 0.98 \\
Convex--concave & Mild & 0.96 & 0.93 & 0.95 & 0.98 & 0.96 & 0.94 & 0.96 & 0.99 \\
Convex--concave & Moderate & 0.94 & 0.92 & 0.96 & 0.97 & 0.94 & 0.92 & 0.96 & 0.98 \\
Convex--concave & Severe & 0.96 & 0.95 & 0.96 & 0.96 & 0.96 & 0.94 & 0.96 & 0.97 \\
Logistic & Mild & 0.96 & 0.93 & 0.96 & 0.98 & 0.96 & 0.93 & 0.97 & 0.99 \\
Logistic & Moderate & 0.96 & 0.94 & 0.96 & 0.98 & 0.96 & 0.96 & 0.96 & 0.99 \\
Logistic & Severe & 0.98 & 0.96 & 0.97 & 0.99 & 0.96 & 0.96 & 0.98 & 0.97 \\
Hung--Swallow & Mild & 0.94 & 0.98 & 0.94 & 0.92 & 0.94 & 0.98 & 0.94 & 0.96 \\
Hung--Swallow & Moderate & 0.95 & 0.94 & 0.98 & 0.97 & 0.96 & 0.93 & 0.98 & 0.95 \\
Hung--Swallow & Severe & 0.98 & 0.97 & 0.98 & 0.99 & 0.96 & 0.96 & 0.98 & 0.95 \\
\bottomrule
\end{tabular}

%% file: tables/mixed_bf.tex
\begin{tabular}{llcccccccccccc}
\toprule
 & & \multicolumn{3}{c}{BADGER, P1} & \multicolumn{2}{c}{BADGER-grid, P1} & \multicolumn{3}{c}{BADGER, P2} & LRT & \multicolumn{3}{c}{Detection at a $5\%$ false-alarm rate} \\
\cmidrule(lr){3-5}\cmidrule(lr){6-7}\cmidrule(lr){8-10}\cmidrule(lr){11-11}\cmidrule(lr){12-14}
Shape & Level & $<1/10$ & $>10$ & median & $<1/10$ & $>10$ & $<1/10$ & $>10$ & median & nominal $5\%$ & $B_{10}$, P1 & $B_{10}$, P2 & LRT \\
\midrule
None & --- & 0.86 & 0.00 & $-$1.43 & 0.90 & 0.00 & 0.00 & 0.00 & $-$0.59 & 0.07 & 0.05 & 0.05 & 0.05 \\
Concave & Mild & 0.72 & 0.01 & $-$1.27 & 0.81 & 0.00 & 0.00 & 0.03 & $-$0.48 & 0.14 & 0.14 & 0.14 & 0.14 \\
Concave & Moderate & 0.19 & 0.10 & $-$0.48 & 0.32 & 0.08 & 0.00 & 0.16 & +0.14 & 0.61 & 0.64 & 0.61 & 0.56 \\
Concave & Severe & 0.03 & 0.47 & +0.91 & 0.04 & 0.36 & 0.00 & 0.60 & +1.30 & 0.93 & 0.94 & 0.94 & 0.92 \\
Linear & Mild & 0.61 & 0.01 & $-$1.19 & 0.78 & 0.01 & 0.00 & 0.03 & $-$0.40 & 0.18 & 0.18 & 0.18 & 0.16 \\
Linear & Moderate & 0.07 & 0.22 & +0.15 & 0.15 & 0.17 & 0.00 & 0.37 & +0.65 & 0.77 & 0.81 & 0.79 & 0.76 \\
Linear & Severe & 0.00 & 0.72 & +2.17 & 0.00 & 0.64 & 0.00 & 0.80 & +2.39 & 0.96 & 0.98 & 0.97 & 0.96 \\
Convex & Mild & 0.59 & 0.01 & $-$1.13 & 0.75 & 0.01 & 0.00 & 0.03 & $-$0.37 & 0.18 & 0.18 & 0.18 & 0.17 \\
Convex & Moderate & 0.04 & 0.26 & +0.27 & 0.10 & 0.17 & 0.00 & 0.40 & +0.77 & 0.81 & 0.84 & 0.84 & 0.80 \\
Convex & Severe & 0.00 & 0.72 & +1.83 & 0.01 & 0.62 & 0.00 & 0.84 & +2.07 & 0.98 & 0.98 & 0.98 & 0.98 \\
Concave--convex & Mild & 0.68 & 0.02 & $-$1.22 & 0.79 & 0.01 & 0.00 & 0.03 & $-$0.44 & 0.17 & 0.16 & 0.16 & 0.16 \\
Concave--convex & Moderate & 0.17 & 0.14 & $-$0.13 & 0.24 & 0.10 & 0.00 & 0.27 & +0.42 & 0.67 & 0.69 & 0.68 & 0.64 \\
Concave--convex & Severe & 0.02 & 0.58 & +1.20 & 0.04 & 0.45 & 0.00 & 0.66 & +1.55 & 0.94 & 0.96 & 0.94 & 0.92 \\
Convex--concave & Mild & 0.58 & 0.01 & $-$1.11 & 0.76 & 0.01 & 0.00 & 0.03 & $-$0.36 & 0.20 & 0.20 & 0.19 & 0.18 \\
Convex--concave & Moderate & 0.05 & 0.33 & +0.41 & 0.07 & 0.23 & 0.00 & 0.44 & +0.90 & 0.83 & 0.86 & 0.85 & 0.81 \\
Convex--concave & Severe & 0.00 & 0.85 & +2.35 & 0.00 & 0.76 & 0.00 & 0.90 & +2.65 & 0.98 & 1.00 & 0.99 & 0.98 \\
Logistic & Mild & 0.55 & 0.02 & $-$1.08 & 0.72 & 0.01 & 0.00 & 0.03 & $-$0.33 & 0.22 & 0.22 & 0.20 & 0.19 \\
Logistic & Moderate & 0.14 & 0.13 & $-$0.13 & 0.22 & 0.10 & 0.00 & 0.28 & +0.42 & 0.70 & 0.74 & 0.70 & 0.67 \\
Logistic & Severe & 0.00 & 0.86 & +2.82 & 0.00 & 0.76 & 0.00 & 0.92 & +2.94 & 0.98 & 1.00 & 0.98 & 0.98 \\
Hung--Swallow & Mild & 0.04 & 0.34 & +0.44 & 0.07 & 0.27 & 0.00 & 0.46 & +0.90 & 0.91 & 0.90 & 0.90 & 0.88 \\
Hung--Swallow & Moderate & 0.00 & 0.86 & +2.77 & 0.00 & 0.79 & 0.00 & 0.92 & +2.96 & 1.00 & 1.00 & 1.00 & 1.00 \\
Hung--Swallow & Severe & 0.00 & 0.96 & +4.20 & 0.00 & 0.92 & 0.00 & 0.98 & +4.16 & 1.00 & 1.00 & 1.00 & 1.00 \\
\bottomrule
\end{tabular}

%% file: tables/mixed_diag.tex
\begin{tabular}{llcccccccccccccc}
\toprule
 & & \multicolumn{6}{c}{BADGER} & \multicolumn{4}{c}{BADGER-grid} & \multicolumn{4}{c}{BADGER-sep} \\
\cmidrule(lr){3-8}\cmidrule(lr){9-12}\cmidrule(lr){13-16}
Shape & Level & $\beta_0$ & $h(1,5)$ & $h(1,10)$ & $\max\widehat R$ & $h(1,10)$, 5\% & $\widehat R>1.05$ & $\beta_0$ & $h(1,5)$ & $h(1,10)$ & $\max\widehat R$ & $\beta_0$ & $h(1,5)$ & $h(1,10)$ & $\max\widehat R$ \\
\midrule
None & --- & 1657 & 771 & 670 & 1.028 & 382 & 0.000 & 1717 & 826 & 829 & 1.017 & 1959 & 1121 & 1524 & 1.009 \\
Concave & Mild & 1426 & 698 & 577 & 1.029 & 325 & 0.000 & 1395 & 742 & 756 & 1.025 & 1717 & 1051 & 1380 & 1.008 \\
Concave & Moderate & 1011 & 571 & 491 & 1.082 & 86 & 0.020 & 965 & 567 & 646 & 1.041 & 1206 & 857 & 1178 & 1.012 \\
Concave & Severe & 785 & 433 & 300 & 1.098 & 63 & 0.045 & 737 & 516 & 659 & 1.045 & 935 & 787 & 1095 & 1.013 \\
Linear & Mild & 1349 & 670 & 571 & 1.029 & 308 & 0.000 & 1388 & 698 & 727 & 1.013 & 1651 & 1009 & 1371 & 1.011 \\
Linear & Moderate & 900 & 518 & 448 & 1.098 & 71 & 0.040 & 832 & 534 & 659 & 1.047 & 1065 & 826 & 1085 & 1.009 \\
Linear & Severe & 579 & 245 & 210 & 1.140 & 67 & 0.060 & 589 & 505 & 677 & 1.036 & 730 & 729 & 1058 & 1.019 \\
Convex & Mild & 1284 & 671 & 569 & 1.037 & 297 & 0.000 & 1288 & 683 & 716 & 1.025 & 1610 & 995 & 1297 & 1.007 \\
Convex & Moderate & 823 & 483 & 463 & 1.088 & 86 & 0.040 & 807 & 555 & 663 & 1.020 & 985 & 807 & 1060 & 1.015 \\
Convex & Severe & 513 & 223 & 295 & 1.084 & 88 & 0.075 & 528 & 486 & 650 & 1.022 & 672 & 708 & 937 & 1.015 \\
Concave--convex & Mild & 1378 & 685 & 577 & 1.063 & 286 & 0.005 & 1373 & 715 & 734 & 1.028 & 1660 & 1017 & 1348 & 1.012 \\
Concave--convex & Moderate & 940 & 556 & 481 & 1.114 & 112 & 0.020 & 888 & 548 & 653 & 1.043 & 1124 & 839 & 1133 & 1.013 \\
Concave--convex & Severe & 694 & 367 & 246 & 1.098 & 61 & 0.090 & 664 & 477 & 659 & 1.036 & 843 & 770 & 1023 & 1.014 \\
Convex--concave & Mild & 1286 & 668 & 570 & 1.016 & 303 & 0.000 & 1311 & 699 & 714 & 1.021 & 1583 & 981 & 1318 & 1.009 \\
Convex--concave & Moderate & 835 & 489 & 448 & 1.090 & 85 & 0.030 & 815 & 548 & 657 & 1.042 & 991 & 816 & 1076 & 1.013 \\
Convex--concave & Severe & 477 & 210 & 266 & 1.100 & 83 & 0.070 & 556 & 498 & 686 & 1.042 & 668 & 739 & 1019 & 1.018 \\
Logistic & Mild & 1271 & 659 & 555 & 1.024 & 280 & 0.000 & 1303 & 670 & 708 & 1.022 & 1594 & 996 & 1324 & 1.011 \\
Logistic & Moderate & 977 & 554 & 416 & 1.104 & 78 & 0.030 & 949 & 542 & 649 & 1.028 & 1173 & 858 & 1152 & 1.011 \\
Logistic & Severe & 562 & 202 & 135 & 1.129 & 62 & 0.115 & 545 & 403 & 640 & 1.040 & 752 & 755 & 1137 & 1.016 \\
Hung--Swallow & Mild & 949 & 512 & 284 & 1.107 & 61 & 0.060 & 790 & 399 & 586 & 1.036 & 1090 & 829 & 1144 & 1.009 \\
Hung--Swallow & Moderate & 599 & 236 & 125 & 1.117 & 56 & 0.150 & 544 & 318 & 535 & 1.069 & 791 & 742 & 1188 & 1.016 \\
Hung--Swallow & Severe & 489 & 186 & 215 & 1.112 & 75 & 0.075 & 420 & 276 & 458 & 1.044 & 611 & 676 & 1062 & 1.028 \\
\bottomrule
\end{tabular}

%% file: tables/mixed_pooledsp.tex
\begin{tabular}{llccccc}
\toprule
Shape & Fit & bias $\widehat\beta_0$ (MC SE) & SD & cover $\beta_0$ & error$\times100$ & $\widehat\varphi$ (cover) \\
\midrule
Concave & BADGER, $\varphi=1-\Sp$ & $-$0.03 (0.01) & 0.13 & 0.96 & 5.0 & --- \\
 & BADGER, $\varphi$ free & \phantom{+}0.00 (0.01) & 0.14 & 0.97 & 5.2 & 0.035 (0.95) \\
 & LOG & $-$0.05 (0.01) & 0.14 & --- & 6.2 & --- \\
 & GT$_0$ & $-$0.20 (0.01) & 0.11 & --- & --- & --- \\
\midrule
Linear & BADGER, $\varphi=1-\Sp$ & $-$0.06 (0.01) & 0.15 & 0.94 & 6.3 & --- \\
 & BADGER, $\varphi$ free & $-$0.03 (0.01) & 0.15 & 0.95 & 5.7 & 0.033 (0.95) \\
 & LOG & $-$0.08 (0.01) & 0.15 & --- & 7.7 & --- \\
 & GT$_0$ & $-$0.28 (0.01) & 0.12 & --- & --- & --- \\
\midrule
Logistic & BADGER, $\varphi=1-\Sp$ & $-$0.01 (0.01) & 0.14 & 0.97 & 5.8 & --- \\
 & BADGER, $\varphi$ free & +0.01 (0.01) & 0.14 & 0.98 & 6.5 & 0.034 (0.95) \\
 & LOG & $-$0.03 (0.01) & 0.14 & --- & 6.3 & --- \\
 & GT$_0$ & $-$0.21 (0.01) & 0.11 & --- & --- & --- \\
\bottomrule
\end{tabular}

%% file: tables/mixed_accuracy.tex
\begin{tabular}{llcccccc}
\toprule
Shape & Accuracy & bias $\widehat\beta_0$ (MC SE) & cover $\beta_0$ & width $\beta_0$ & error$\times100$ & bias $\widehat\Se$ / $\widehat\Sp$ & cover $\Se$ / $\Sp$ \\
\midrule
Concave & known, correct & +0.01 (0.01) & 0.94 & 0.61 & 5.2 & --- & --- \\
 & known, 0.99/0.99 & +0.07 (0.01) & 0.90 & 0.58 & 6.0 & --- & --- \\
 & centered at truth, $n_0=100$ & +0.04 (0.01) & 0.96 & 0.75 & 5.8 & +0.003 / +0.002 & 1.00 / 1.00 \\
 & centered at truth, $n_0=20$ & +0.07 (0.01) & 0.94 & 0.81 & 6.8 & +0.011 / +0.005 & 1.00 / 0.90 \\
 & centered at 0.99/0.99, $n_0=100$ & +0.04 (0.01) & 0.96 & 0.76 & 6.6 & +0.039 / +0.006 & 0.07 / 0.90 \\
 & centered at 0.99/0.99, $n_0=20$ & +0.07 (0.01) & 0.94 & 0.82 & 7.2 & +0.038 / +0.007 & 0.72 / 0.84 \\
 & as $n_0=20$, pools of 5 only & +0.12 (0.01) & 0.95 & 0.78 & --- & +0.009 / +0.007 & 1.00 / 0.94 \\
\midrule
Linear & known, correct & $-$0.02 (0.01) & 0.93 & 0.65 & 5.2 & --- & --- \\
 & known, 0.99/0.99 & +0.04 (0.01) & 0.93 & 0.61 & 6.0 & --- & --- \\
 & centered at truth, $n_0=100$ & +0.01 (0.01) & 0.98 & 0.79 & 5.6 & +0.002 / +0.002 & 1.00 / 1.00 \\
 & centered at truth, $n_0=20$ & +0.05 (0.01) & 0.97 & 0.86 & 6.4 & +0.009 / +0.005 & 1.00 / 0.94 \\
 & centered at 0.99/0.99, $n_0=100$ & +0.02 (0.01) & 0.97 & 0.80 & 6.4 & +0.039 / +0.006 & 0.07 / 0.92 \\
 & centered at 0.99/0.99, $n_0=20$ & +0.05 (0.01) & 0.95 & 0.86 & 6.8 & +0.037 / +0.008 & 0.78 / 0.86 \\
 & as $n_0=20$, pools of 5 only & +0.08 (0.01) & 0.98 & 0.83 & --- & +0.008 / +0.006 & 1.00 / 0.94 \\
\midrule
Logistic & known, correct & $-$0.08 (0.02) & 0.96 & 0.91 & 6.7 & --- & --- \\
 & known, 0.99/0.99 & $-$0.04 (0.02) & 0.94 & 0.83 & 6.7 & --- & --- \\
 & centered at truth, $n_0=100$ & $-$0.08 (0.02) & 0.98 & 1.14 & 6.5 & +0.001 / $-$0.001 & 1.00 / 1.00 \\
 & centered at truth, $n_0=20$ & $-$0.04 (0.02) & 0.98 & 1.25 & 6.5 & +0.005 / +0.002 & 1.00 / 0.98 \\
 & centered at 0.99/0.99, $n_0=100$ & $-$0.03 (0.02) & 0.98 & 1.10 & 6.6 & +0.039 / +0.006 & 0.00 / 0.98 \\
 & centered at 0.99/0.99, $n_0=20$ & $-$0.02 (0.02) & 0.98 & 1.22 & 6.5 & +0.038 / +0.006 & 0.86 / 0.91 \\
 & as $n_0=20$, pools of 5 only & $-$0.04 (0.02) & 0.98 & 1.22 & --- & +0.003 / \phantom{+}0.000 & 1.00 / 1.00 \\
\bottomrule
\end{tabular}

%% file: tables/prior_implied.tex
\begin{tabular}{llcccccccc}
\toprule
 & & & \multicolumn{3}{c}{Quantiles} & \multicolumn{4}{c}{Prior probability that the value is} \\
\cmidrule(lr){4-6}\cmidrule(lr){7-10}
Prior & Value & Mean & $25\%$ & $50\%$ & $75\%$ & $<0.1$ & $<0.5$ & $>0.8$ & $>0.9$ \\
\midrule
P1, grid & $h(1,10)$ & 0.11 & 0.01 & 0.01 & 0.07 & 0.78 & 0.92 & 0.03 & 0.01 \\
 & $h(1,5)$ & 0.21 & 0.01 & 0.05 & 0.32 & 0.58 & 0.83 & 0.07 & 0.04 \\
P1, Bernstein & $h(1,10)$ & 0.23 & 0.04 & 0.13 & 0.36 & 0.42 & 0.86 & 0.05 & 0.02 \\
 & $h(1,5)$ & 0.32 & 0.10 & 0.25 & 0.49 & 0.25 & 0.75 & 0.06 & 0.03 \\
P2, Bernstein & $h(1,10)$ & 0.55 & 0.25 & 0.56 & 0.86 & 0.11 & 0.45 & 0.31 & 0.21 \\
 & $h(1,5)$ & 0.60 & 0.34 & 0.65 & 0.89 & 0.07 & 0.38 & 0.35 & 0.23 \\
\bottomrule
\end{tabular}

%% file: tables/mixed_prior.tex
\begin{tabular}{llcccccccc}
\toprule
 & & \multicolumn{4}{c}{P1: $a=1/L$ (default)} & \multicolumn{4}{c}{P2: asymmetric} \\
\cmidrule(lr){3-6}\cmidrule(lr){7-10}
Shape & Level & bias & RMSE & cover & error & bias & RMSE & cover & error \\
\midrule
None & --- & +0.07 & 0.13 & 0.92 & 7.0 & +0.05 & 0.11 & 0.95 & 5.0 \\
Concave & Mild & +0.05 & 0.13 & 0.94 & 4.9 & +0.02 & 0.12 & 0.96 & 3.4 \\
Concave & Moderate & +0.02 & 0.13 & 0.96 & 5.5 & \phantom{+}0.00 & 0.14 & 0.96 & 5.9 \\
Concave & Severe & +0.01 & 0.15 & 0.99 & 6.0 & $-$0.01 & 0.17 & 0.98 & 6.4 \\
Linear & Mild & +0.04 & 0.13 & 0.96 & 4.4 & +0.01 & 0.12 & 0.96 & 3.6 \\
Linear & Moderate & \phantom{+}0.00 & 0.15 & 0.95 & 5.9 & $-$0.02 & 0.16 & 0.94 & 6.7 \\
Linear & Severe & $-$0.06 & 0.18 & 0.98 & 6.9 & $-$0.07 & 0.19 & 0.97 & 7.3 \\
Convex & Mild & +0.03 & 0.13 & 0.96 & 4.2 & \phantom{+}0.00 & 0.13 & 0.96 & 3.7 \\
Convex & Moderate & $-$0.04 & 0.15 & 0.96 & 6.3 & $-$0.06 & 0.17 & 0.96 & 7.4 \\
Convex & Severe & $-$0.12 & 0.22 & 0.95 & 8.8 & $-$0.13 & 0.24 & 0.96 & 9.3 \\
Concave--convex & Mild & +0.04 & 0.13 & 0.95 & 4.6 & +0.02 & 0.12 & 0.96 & 3.5 \\
Concave--convex & Moderate & \phantom{+}0.00 & 0.14 & 0.96 & 6.1 & $-$0.02 & 0.15 & 0.94 & 6.7 \\
Concave--convex & Severe & $-$0.05 & 0.16 & 0.96 & 6.3 & $-$0.07 & 0.18 & 0.96 & 7.1 \\
Convex--concave & Mild & +0.03 & 0.13 & 0.96 & 4.3 & \phantom{+}0.00 & 0.13 & 0.96 & 3.7 \\
Convex--concave & Moderate & $-$0.01 & 0.16 & 0.94 & 6.0 & $-$0.03 & 0.17 & 0.95 & 6.9 \\
Convex--concave & Severe & $-$0.06 & 0.19 & 0.96 & 7.0 & $-$0.06 & 0.20 & 0.96 & 7.5 \\
Logistic & Mild & +0.04 & 0.13 & 0.96 & 4.7 & +0.02 & 0.12 & 0.96 & 3.8 \\
Logistic & Moderate & +0.03 & 0.14 & 0.96 & 5.8 & +0.01 & 0.14 & 0.96 & 5.9 \\
Logistic & Severe & +0.02 & 0.16 & 0.98 & 5.5 & +0.02 & 0.17 & 0.95 & 6.3 \\
Hung--Swallow & Mild & +0.06 & 0.15 & 0.94 & 7.3 & +0.04 & 0.15 & 0.95 & 7.2 \\
Hung--Swallow & Moderate & +0.05 & 0.15 & 0.95 & 6.1 & +0.05 & 0.16 & 0.98 & 6.9 \\
Hung--Swallow & Severe & +0.01 & 0.16 & 0.98 & 5.3 & +0.02 & 0.18 & 0.97 & 6.8 \\
\bottomrule
\end{tabular}